\documentclass[a4paper]{article}
\usepackage{textcomp}
\usepackage{authblk}
\usepackage{hyperref}
\usepackage{cite}
\usepackage{tikz}
\usetikzlibrary{patterns,snakes}
\usepackage{mathtools}
\usepackage{dsfont}
\usepackage{tikz}
\usepackage{graphicx}
\usepackage{dcolumn}
\usepackage{bm}

\usepackage{pifont}
\expandafter\let\csname equation*\endcsname\relax
\expandafter\let\csname endequation*\endcsname\relax
\usepackage{amssymb}
\usepackage{amsmath}
\usepackage{amsthm,amsfonts}

\usepackage{epstopdf}
\usepackage{caption}
\usepackage{subcaption}
\usepackage{float}
\usepackage{enumitem}
\usepackage{calc}

\RequirePackage{graphicx}

\usepackage[margin=1in]{geometry}

\theoremstyle{plain}

\newtheorem{theorem}{Theorem}[section]

\newtheorem{conj}[theorem]{Conjecture}

\newtheorem{lemma}[theorem]{Lemma}

\newtheorem{cor}[theorem]{Corollary}

\newcounter{mycounter}
\newcommand{\particles}[1]{
		\foreach \i in {0,...,11}
	{
		\draw [very thick] (\i+0.1,0) -- (\i+0.9,0);
	}
  \setcounter{mycounter}{1}
    \foreach \i in #1
    {
    \ifnum \i = 0
      \draw [very thick, dashed] (\themycounter-0.5,0.4) circle (8pt);
    \else
     \ifnum \i = 1
      \draw [very thick] (\themycounter-0.5,0.4) circle (8pt);
      \else
     \ifnum \i = 21
      \draw [very thick, red] (\themycounter-0.5,0.4) circle (8pt);
       \else
     \ifnum \i = 20
      \draw [very thick, dashed,red] (\themycounter-0.5,0.4) circle (8pt);
    \fi
    \fi
    \fi
    \fi
    \stepcounter{mycounter};
    }
}

\title{Single impurity in the Totally Asymmetric Simple Exclusion Process}

\date{} 

\author[1]{Luigi Cantini\thanks{\href{mailto:luigi.cantini@cyu.fr}{luigi.lantini@cyu.fr}}}
\author[1,2]{Ali Zahra\thanks{\href{mailto:ali.zahra@cyu.fr}{ali.zahra@cyu.fr}}}
\affil[1]{Laboratoire de Physique Th\'eorique et Mod\'elisation, CNRS UMR 8089, CY Cergy Paris Universit\'e, 95302 Cergy-Pontoise Cedex, France}
\affil[2]{
Departamento de Matem\'atica, Instituto Superior T\'ecnico - Av. Rovisco Pais, 1049-001 Lisboa, Portugal}
\begin{document}
\maketitle

\begin{abstract}
We examine the behavior of a single impurity particle embedded within a 
Totally Asymmetric Simple Exclusion Process (TASEP). 
By analyzing the impurity's dynamics, characterized by two arbitrary hopping parameters $ \alpha $ and $\beta$, 
we investigate both its macroscopic impact on the system and its individual trajectory, providing new insights into the interaction between the impurity and the TASEP environment.
We classify the induced hydrodynamic limit shapes based on the initial densities to the left and right of the impurity,  along with the values of
the parameters $\alpha$,$\beta$. 
We develop a new method that enables the analysis of the 
impurity's behavior within an arbitrary density field, thereby generalizing 
the traditional coupling technique used for second-class particles. With 
this tool, we extend to the impurity case under certain parameter conditions, Ferrari and Kipnis's results on the distribution of the 
asymptotic speed of a second-class particle within a rarefaction fan.
\end{abstract}

\section*{Introduction}
The study of non-equilibrium stochastic processes is pivotal in 
understanding various natural phenomena, encompassing a wide range of 
fields including physics, chemistry \cite{van1992stochastic}, biology 
\cite{edelstein2005mathematical,ditlevsen2013introduction}, and 
interdisciplinary areas \cite{blythe2007stochastic}. 
Non-equilibrium statistical mechanics exhibits a 
surprisingly rich array of physical phenomena, including phase transitions 
\cite{privman1997nonequilibrium}, even in one--dimensional 
systems. Models in one spatial dimension often possess the 
advantageous feature of being analytically tractable. 
A paradigmatic example of exactly solvable model is the Asymmetric Simple 
Exclusion Process 
(ASEP), which is widely recognized as a fundamental model for 
non--equilibrium transport processes \cite{chou2011non}. ASEP has found 
applications in diverse contexts such as traffic and pedestrian flow  
\cite{chowdhury2005physics} , mRNA translation by ribosomes 
\cite{ciandrini2010role}, and motor protein transport along single 
filaments \cite{chowdhury2005physics}, among others. 
At a macroscopic level, ASEP's hydrodynamic behavior is well--established 
\cite{rost1981non,rezakhanlou1991hydrodynamic}. However, introducing even a 
single microscopic impurity can substantially impact the system's overall 
behavior and properties, creating complex effects that remain much less 
understood.
One example of an impurity that has received considerable interest in the literature is a slow bond at a single site 
\cite{kolomeisky1998asymmetric,basu2017invariant}. 
In this article, we explore a different type of impurity: an impurity particle (hereafter referred to as the "\emph{impurity}"), which follows distinct dynamics compared to regular particles (hereafter referred to simply as "particles").
To be more specific, we consider the 
Totally Asymmetric Simple 
Exclusion Process (TASEP) in presence of an impurity. In this model 
each site of a one--dimensional lattice 
can be either empty or occupied  by either a particle or  an impurity.
%
%
A particle hops forward on an empty site with rate $1$, while an impurity hops on an empty site with rate $\alpha$.
In addition to this, a particle can overtake an impurity with rate $\beta$. 
Denoting a particle by $\bullet$, an empty site by $\circ$ and an impurity by
$\ast$, 
the dynamics of the model can be schematically represented by the following exchanges with the corresponding rates
$$
 \bullet\, \circ \xrightarrow{1} \circ \bullet\qquad
\ast\,\circ \xrightarrow{\alpha} \circ \ast \qquad
\bullet\, \ast \xrightarrow{\beta} \ast \bullet.
$$
This model has  
been first considered in \cite{derrida1996statphys,mallick1996shocks},
where the stationary measure on a finite periodic lattice was written
in a matrix product ansatz form \cite{derrida1993exact, evans2009matrix}. 
Generalizations with an arbitrary number of impurities have considered in  \cite{cantini2008algebraic,cantini2022hydrodynamic}.

%
When $\alpha = \beta = 1$,  the impurity reduces to what is commonly known as a "\emph{second-class 
particle}", a concept which was introduced as early as the exclusion process itself 
\cite{liggett1976coupling}.
From a hydrodynamic perspective, the evolution of the particle density in this 
setting is governed by the inviscid Burgers equation. 
Here, a second-class particle acts as a tracer, moving along the characteristics of the inviscid Burgers equation that originate from its position
\cite{ferrari1992shocks, 
rezakhanlou1995microscopic}. 
In the scenarios involving a shock, a second-class particle 
accurately tracks the shock's position \cite{andjel1988shocks, 
derrida1993exact}. Interestingly, if such a particle is added to a point where multiple characteristics emerge (i.e., a decreasing discontinuity), 
it will randomly select one of the available characteristics.
The distribution of the asymptotic average speed of the second--class particle is highly sensitive to its precise initial position.
Ferrari and Knipsis \cite{ferrari1995second} have considered the case where in the initial 
configuration the first--class particles are 
distributed following product Bernoulli measure with densities
$\rho_L > \rho_R$ to the left and right of the second--class 
particle respectively. They have shown that the asymptotic speed of 
the second--class particle is uniformly distributed distributed in the interval
$[1-2\rho_L, 1-2\rho_R]$ (see also \cite{mountford2005motion}).
This result has been extended to the ASEP case, as proven in 
\cite{ferrari2009collision}; see also the recent work in 
\cite{aggarwal2023asep}.
More general initial conditions in the TASEP case have been considered by Cator and 
Pimentel \cite{cator2013busemann}. Using the well--known 
last--passage percolation (LPP) formulation of the TASEP, 
and then the connection between second class 
particles and competition interfaces in the LPP 
\cite{ferrari2005competition,ferrari2006roughening}, they have shown 
that the limiting speed of the second--class particle can be 
expressed in terms of the supremum of certain random walks. 

%
%
%

\vspace{.3cm}

While a second-class particle does not alter the macroscopic evolution of 
the first-class particle density, this may no longer hold for an impurity 
with arbitrary values of the parameters $\alpha$ and $\beta$. The first 
objective of this paper is to investigate the impact of a single impurity 
on the density profile, considering the impurity rates $\alpha$ and 
$\beta$. Using heuristic hydrodynamics, we analyze the evolution of an 
initially smooth profile and then focus on cases where the impurity is 
positioned at a density discontinuity, identifying conditions under which 
the impurity modifies the density profile and detailing the nature of these 
changes.

More specifically, we examine an initial profile consisting of a region 
with constant density $\rho_L$ to the left of the impurity and a region with 
constant density $\rho_R$ to the right. When $\rho_L > \rho_R$, the absence of the 
impurity would  lead to the development of a rarefaction fan.
This fan structure generally persists, except in the parameter region 
defined by $\alpha + \beta < 1$, $\beta < \rho_R$, and $\alpha < 1 - \rho_L$. 
Under these conditions, the impurity induces an anti--shock, a 
discontinuity where the density increases from left to right, with a left 
density of $1 - \alpha$ and a right density of $\beta$. 
%
%
%
On top of the anti--shock we may observe a shock 
to the left and or to the right of the impurity. 
 These results are 
summarized in Figure \ref{zoo-fan}.  For numerical simulations illustrating this behavior, see \ref{fig:munu1a}--\ref{fig:munu1d}.
When $\rho_L < \rho_R$, the absence of the impurity would result in a shock 
propagating with speed $v_s=1-\rho_L-\rho_R$.
This shock structure generally remains intact, except in the parameter 
region defined by
$\beta < \rho_R$, and $\alpha < 1 - \rho_L$.
Under these conditions, the shock splits into two shocks, separated by an anti-shock located at the impurity position. For numerical simulations illustrating this behavior, see Figures \ref{shocksa}--\ref{shocksd}.

A second objective of this paper is to analyze the trajectory of the 
impurity itself. When the initial density profile is smooth, the impurity 
propagates at a constant average speed over long times, determined entirely 
by the initial conditions. However, if the impurity is positioned at a 
decreasing density discontinuity with $\alpha+\beta>1$, the situation 
changes: while the impurity still reaches a constant asymptotic speed, this 
speed becomes a random variable, no longer uniquely determined by the 
initial conditions.
The main challenge in extending the analysis of a second-class particle's 
trajectory to that of an impurity with arbitrary jumping rates
$\alpha$ and $\beta$ 
arises from the breakdown of the standard approach, which characterizes the second-class particle as a discrepancy between two coupled TASEP systems
\cite{liggett1976coupling,ferrari1995second}. 
To address this difficulty, we use a \emph{hole--particle} pair 
representation of the impurity, as introduced by Ferrari and Pimentel in the case of a second--class particle \cite{ferrari2005competition}. 
To adapt this approach to the impurity case, we apply Weber's interchangeability theorem from queueing theory 
\cite{weber1979interchangeability}.
This allows us to represent the trajectory of the impurity alone (rather than the entire system) through the dynamics of a particle--hole pair in a TASEP system, where a tagged particle and a tagged hole move with respective jumping rates $\alpha$ and $\beta$.

We are thus able to generalize the result by Ferrari and Kipnis for the 
case where $\alpha,\beta \leq 1$ and $\alpha+\beta\geq 1$ and the impurity
is initially positioned at the interface between a fully occupied region on 
the left and an empty region on the right. 
We prove that, under these conditions, the asymptotic speed of the impurity 
is uniformly distributed over the interval $[1-2\beta, 2\alpha-1]$.
If $\alpha>1$, the impurity may ''escape'' to the right of the rarefaction 
fan; similarly, if $\beta>1$ the impurity may escape to the left.
We conjecture that, conditioned to the impurity to remain within the fan,
its asymptotic speed is uniformly distributed in the interval
 $[\max(1-2\beta,-1), \min(2\alpha-1,1)]$. 
Numerical simulations suggest that, for more general left/right density configurations, the asymptotic speed of the impurity is no longer uniformly distributed within the allowed interval.

The paper is organized as follows.
In Section \ref{Density profile}, we utilize hydrodynamics to investigate 
the behavior of the density profile as a function of the impurity rates $
\alpha$  and $\beta$.
Section \ref{Asymptotic} begins with a review of the \emph{hole--particle} 
pair representation of the impurity. In Section \ref{section:dynamics}, we 
establish one of the paper's main results, Theorem \ref{first:Theorem}, 
which states that the trajectory of the impurity shares the same 
distribution as that of a hole--particle pair in a TASEP with a tagged 
particle having a jumping rate $\alpha$ and a tagged hole with jumping rate 
$\beta$.
In Section \ref{sect:speed-fan}, we apply Theorem \ref{first:Theorem}, to 
first recover, using purely probabilistic arguments, the known result 
\cite{mallick1996shocks,derrida1999bethe} regarding the asymptotic 
speed of an impurity in an initially uniform density profile. 
We then apply the same theorem to determine the distribution of the speed 
of an impurity with jumping rates $\alpha,\beta<1$ and $\alpha+\beta>1$,  
initially locate at the interface of a $1$--$0$ density profile.
Finally, we conclude with a computation of the escape probabilities for the 
impurity from the rarefaction fan, starting from an initial $1$--$0$ density 
profile.

\section{Hydrodynamics in presence of a single impurity}
\label{Density profile}

The behavior of a system with a single impurity in an initially uniform 
profile of particles with density $\rho$ was analyzed in 
\cite{mallick1996shocks}. 
It was found that the impurity moves at a constant asymptotic  speed
$v_\ast$, which is given by the following compact expression:
\begin{equation}\label{def-speed}
v_\ast(\rho) = \min(\alpha,1-\rho)- \min(\beta,\rho). 
\end{equation}
This result was derived in \cite{mallick1996shocks} 
using the matrix product ansatz to represent the stationary measure of 
the model under periodic boundary conditions. 
A different derivation was presented in \cite{derrida1999bethe}, 
employing Bethe ansatz techniques. In Section \ref{sect:speed-fan}, 
we will provide a purely probabilistic derivation of this result.

\subsection{Constant or smooth initial density profile}

As shown in \cite{mallick1996shocks}, the presence of an impurity in an initially uniform profile of density $\rho_0$
has no macroscopic impact on the particle density unless
$\beta < \rho_0 < 1-\alpha$. 
Within this range, the impurity induces an \emph{anti-shock}, i.e.
an increasing discontinuity in the density profile at its location, with density
$\beta$ to its right greater than the density $1-\alpha$ to its left.
This anti--shock moves at constant speed $\alpha-\beta$, which corresponds to the speed of the impurity  $v_\ast$.
%
To determine the complete evolution of the density profile, we assume that, aside from the position of the impurity, the density profile follows the Burgers equation:
\begin{equation}\label{burgers}
\partial_t \rho + \bar{v}(\rho) \partial_x \rho=0,
\end{equation}   
where the characteristic speed $\bar{v}(\rho)=1-2\rho$ represents the 
propagation velocity of a patch of local density $\rho$.
Consider the domain to the left of the impurity. Here, as described 
above, the density profile is constrained to a value of $1-\alpha$ just to the left of the impurity, while it reaches density $\rho_0$ as $x\rightarrow -\infty$. Solving the Burgers equation with these boundary conditions, we find that the profile develops a shock, transitioning between density $\rho_0$ on the left and $1-\alpha$ on the right.
This shock propagates at a speed of $\alpha-\rho_0$, which is, as expected, slower than $v_\ast$. Similarly, on the right side of the impurity, we observe the formation of a second shock between a density
$\beta$ on the left and $\rho_0$ on the right. This shock moves at a speed of
 $1-\rho_0-\beta$, which is greater than $v_\ast$. 
%
The table below summarizes the results regarding the impurity speed and the occurrence of a density discontinuity.  

\begin{center}
\renewcommand{\arraystretch}{1.5}
\renewcommand{\tabcolsep}{0.4cm}
\begin{tabular}{|c|c|c|c|}
\hline
&$v_\ast(\rho_0)$ & $\rho_0$ & Discontinuity\\
\hline 
(a)&$1-2\rho_0$ & $\beta\geq \rho_0 \geq 1-\alpha$ & NO\\
(b)& $1-\rho_0-\beta$ & $\max (\beta,1-\alpha) \leq \rho_0 $ & NO \\
(c)&$\alpha-\rho_0$ & $\min (\beta,1-\alpha) \geq \rho_0 $ & NO \\
(d)&$\alpha-\beta$ & $\beta\leq \rho_0 \leq 1-\alpha $ & YES\\
\hline
\end{tabular}
\end{center}

\vspace{.2cm}
\noindent
Suppose now that the initial density profile is no longer uniform but smooth.
To analyze the situation, we assume local equilibrium, which means
that the behavior of the impurity eventually depends only on the local  
density in its vicinity.
Let's distinguish the case $\alpha+\beta$ greater or smaller than $1$.

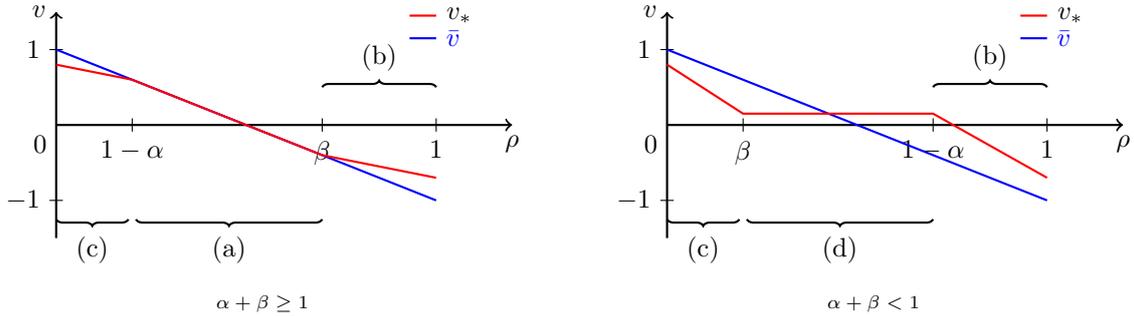
\begin{figure}[H]
	\begin{subfigure}{.5\textwidth}
		\centering
		\begin{tikzpicture}[scale = 0.5]
	\draw[thick,->] (0,0) -- (12,0) node[below] {$\rho$};
	\draw[thick,->] (0,-3) -- (0,3) node[left] {$v$};
	\draw (0cm+5pt,2) -- ( 0cm-5pt,2) node[left] {$1$};	
	\draw (0cm+5pt,-2) -- ( 0cm-5pt,-2) node[left] {$-1$};	
	
		\draw (7 cm,5pt) -- (7 cm,-5pt) node[below ] {$\beta$};
	
	
	\draw (10 cm,5pt) -- (10 cm,-5pt) node[below] {$1$};
 	\draw (0 cm,1pt) -- (0 cm,-1pt) node[below left] {$0$};
   	\draw (2 cm,5pt) -- (2 cm,-5pt) node[below] {$1-\alpha$};

        \draw [thick, blue] (0,2) -- (10,-2);
        
        \draw [thick,red] (0,1.6) -- (2,1.2) -- (7,-0.4*2) -- (10,-0.7*2);

\draw [
    thick,
    decoration={
        brace,
        mirror,
        raise=1cm
    },
    decorate
] (0,-0.5) -- (1.9,-0.5)
node [pos=0.5,anchor=north,yshift=-1.1cm] {(c)};

\draw [
    thick,
    decoration={
        brace,
        mirror,
        raise=1cm
    },
    decorate
] (2.1,-0.5) -- (7,-0.5)
node [pos=0.5,anchor=north,yshift=-1.1cm] {(a)};

\draw [
    thick,
    decoration={
        brace,
        raise=0cm
    },
    decorate
] (7,1) -- (10,1)
node [pos=0.5,anchor=north,yshift=0.7cm] {(b)};
\draw [thick, blue](9.3,2.3) -- (10,2.3) node[anchor= west] {$\bar{v}$};
\draw [thick,red] (9.3,2.9) -- (10,2.9);
\draw [very thick] (10,2.9) node[anchor= west] {$v_*$};
\end{tikzpicture}
		\caption*{\scriptsize $\alpha+\beta \geq 1$}
		\label{fig:sub1}
	\end{subfigure}
		\begin{subfigure}{.5\textwidth}
		\centering
		\begin{tikzpicture}[scale = 0.5]
	\draw[thick,->] (0,0) -- (12,0) node[below] {$\rho$};
	\draw[thick,->] (0,-3) -- (0,3) node[left] {$v$};
	\draw (0cm+5pt,2) -- ( 0cm-5pt,2) node[left] {$1$};	
	\draw (0cm+5pt,-2) -- ( 0cm-5pt,-2) node[left] {$-1$};	
	
	\draw (7 cm,5pt) -- (7 cm,-5pt) node[below ] {$1-\alpha$};
	
	
	\draw (10 cm,5pt) -- (10 cm,-5pt) node[below] {$1$};
 	\draw (0 cm,1pt) -- (0 cm,-1pt) node[below left] {$0$};
   	\draw (2 cm,5pt) -- (2 cm,-5pt) node[below] {$\beta$};

        \draw [thick, blue] (0,2) -- (10,-2);
        
        \draw [thick,red] (0,1.6) -- (2,0.3) -- (7,0.3) -- (10,-0.7*2);

\draw [
    thick,
    decoration={
        brace,
        mirror,
        raise=1cm
    },
    decorate
] (0,-0.5) -- (1.9,-0.5)
node [pos=0.5,anchor=north,yshift=-1.1cm] {(c)};

\draw [
    thick,
    decoration={
        brace,
        mirror,
        raise=1cm
    },
    decorate
] (2.1,-0.5) -- (7,-0.5)
node [pos=0.5,anchor=north,yshift=-1.1cm] {(d)};

\draw [
    thick,
    decoration={
        brace,
        raise=0cm
    },
    decorate
] (7,1) -- (10,1)
node [pos=0.5,anchor=north,yshift=0.7cm] {(b)};
\draw [thick, blue](9.3,2.3) -- (10,2.3) node[anchor= west] {$\bar{v}$};

\draw [thick,red] (9.3,2.9) -- (10,2.9);
\draw [very thick] (10,2.9) node[anchor= west] {$v_*$};
\end{tikzpicture}
		\caption*{\scriptsize $\alpha+\beta<1$}
		\label{fig:sub2}
	\end{subfigure}

\caption{Plots of the impurity speed $v_\ast$ (red line) and the speed of the characteristics of the Burger's equation $\bar{v}$ (blue) for the cases $\alpha+\beta\geq 1$ (left) and $\alpha+\beta<1$ (right).}\label{fig:speeds}
\end{figure}

\vspace{.2cm}
\noindent
In the case $ \alpha+\beta\geq 1$, the parameter region (d) 
does not exist. 
This implies that the particle density is unaffected by the impurity and 
therefore satisfies the Burgers equation everywhere.
The impurity then moves within an independently evolving density profile.   
When the initial surrounding density falls into region (a)
we have $v_\ast=\bar{v}$ (see Fig.\ref{fig:speeds}, left)  meaning that 
the impurity moves together with the patch of density $\rho$.
If the initial surrounding density falls into region
(b) the impurity moves faster than the density patch  and continues to do so unless  
it encounters a region of density $\beta$, at which point 
it stabilizes at a patch of this density, moving at speed 
$v_\ast=\bar{v}(\beta)=1-2\beta$. Similarly,  when the 
initial density surrounding the impurity falls into region 
(c),  the impurity moves slower than the density patch 
unless it encounters a region of density $1-\alpha$, at 
which points it stabilizes at a patch of this density, 
moving at speed $v_\ast=\bar{v}(1-\alpha)=2\alpha-1$.

\vspace{.2cm}
\noindent
The situation becomes more intricate when $\alpha+\beta<1$, as the presence of the impurity can significantly alter the density profile. 
This alteration arises from the potential formation of an anti--shock.
This phenomenon occurs if the initial surrounding density falls into 
region (d). In such instances, as explained above in the case of an 
initial uniform density profile, an anti--shock emerges between densities 
$1-\alpha$ to the left of the impurity and $\beta$ to its right.
Depending on the whole initial profile, a shock may manifest behind 
the impurity, in front of it, or even both simultaneously.

\subsection{$1-0$ step initial density profile}
\label{sect-0-1}

The preceding discussion focused solely on initially smooth density 
profiles. It is intriguing to explore the system's behavior when the 
impurity is initially positioned at a discontinuity.  
It is instructive to first examine the scenario where all sites to the left of the impurity are occupied by particles, while all sites to its right are empty. Later, we will delve into the case of arbitrary initial step profiles.

\subsubsection*{Case $\alpha+\beta\geq 1$}

In this case we know that the impurity has no impact on the evolution of
the density profile. Suppose that the initial  step is located at the origin. We know that it will evolve into a rarefaction fan of equation $\rho(x;t) = u(x/t)$ with
\begin{equation}
u(v) = 
\left\{
\begin{array}{cc}
1 & v\leq -1\\
\frac{1-v}{2} & -1 \leq v\leq 1\\
0 & v\geq 1
\end{array}
\right.
\end{equation}
Once the fan has formed, the density profile is smooth, therefore the later 
evolution of the impurity is fully determined by the 
local density surrounding the impurity.
Let's enumerate the different possibilities that can occur, $\rho$ being 
the density surrounding the impurity:
\begin{itemize}[nosep,
  align=left,
  leftmargin=0pt,
  labelwidth=1.25em,
  itemindent=1.25em,
  labelsep=0pt]
\item[$\blacktriangleright$] If $\alpha>1$ and the impurity is to the right of all the 
particles, it may happen that no particle will be able to overcome it. In 
this  we say that the impurity \emph{escapes the fan from the right} an it will move with speed $v_\ast=\alpha$.
\item[$\blacktriangleright$] If $\beta>1$ and the impurity is to the left of all the 
empty sites, it may happen that impurity will not be able to overcome any empty site. In 
this  we say that the impurity \emph{escapes the fan from the left} an it will move with speed $v_\ast=-\beta$.
\item[$\blacktriangleright$] $\rho$ in region (b). This is possible only if $\beta<1$. In this 
case 
the impurity moves faster than the surrounding density, causing the latter to gradually decrease over time until it matches $\beta$. At this point, the impurity's speed becomes equal to that of the surrounding density patch, i.e. $v_\ast = 1 - 2\beta$.
\item[$\blacktriangleright$] $\rho$ in region (c). This is possible only if $\alpha<1$. In this case the impurity moves slower than the surrounding density, causing the latter to gradually decrease over time
until it matches $1-\alpha$. At that moment the impurity's speed  becomes equal to the speed of the surrounding density patch, i.e. $v_\ast = 2\alpha-1$.
\item[$\blacktriangleright$] $\rho$ in region (a) and $\rho\neq 0,1$: 
the impurity moves in tandem with the density patch $\rho$ at a speed of $v_\ast = 1 - 2\rho$.
\end{itemize}
In conclusion, for the asymptotic speed of the impurity we have
the following possibilities
\begin{equation}
\begin{split}
&v_\ast \in [\max(1-2\beta,-1), \min(2\alpha-1,1)]\\
 &v_\ast=\alpha\quad  \text{possible only if } \alpha>1\\
&v_\ast=-\beta\quad  \text{possible only if } \beta>1.
\end{split}
\end{equation} 
The initial configuration doesn't allow for a definitive prediction of the outcome; each possibility will occur with a certain probability, which strongly depends on the exact initial position of the impurity.
In Section \ref{sect:esc} we will compute, $P_R$ and $P_L$, which are the probabilities of the impurity escaping to the right and left, respectively. For the case where the initial configuration has all the sites to the left of the impurity occupied and all those to the right empty, we obtain:
\begin{equation}\label{esc-prob}
\begin{split}
P_R&=\frac{\alpha-1}{\alpha+\beta-1}\\
P_L&=\frac{\beta-1}{\alpha+\beta-1}.
\end{split}
\end{equation}
Results for generic position of the impurity are found in Section
\ref{sect:esc}.
%
%
%
%
%

\subsubsection*{Case $\alpha+\beta< 1$}

Since $ \beta < 1-\alpha $, only the regions (b), (c) and (d) of the 
density are meaningful. If the impurity finds itself in the region (b), 
then it moves at a speed greater than the speed of the surrounding 
density profile and it will eventually reach
the region (d).
Similarly if the impurity finds itself in region (c), then it moves at a 
speed smaller than the speed of the surrounding density profile and it 
will eventually reach the region (d) as well.
After getting  in region (d), the particle will move at a constant speed $ v_\ast= \alpha - \beta $ and it will form of an anti--shock, the density to the right and to the left of the impurity being respectively equal to $\beta$  and $1-\alpha$.
In the long run, this anti--shock combines with the fan to form 
a density profile described by $\rho(x,t)=u(x/t)$, with
\begin{equation}
u(v) = 
\left\{
\begin{array}{ll}
1  &   v < -1 \\

\frac{1}{2}(1-v)
 &  -1 \leq v  \leq ( 2\alpha - 1) ~~~\text{or}~~~ (1-2\beta) \leq v  \leq 1\\

1-\alpha  & (2\alpha - 1) \leq v < \alpha  - \beta    \\

\beta  &  \alpha  - \beta < v \leq (1-2\beta )  \\

0  &     v >1

\end{array}
\right.
\end{equation}
The late time density profile $\rho(x,t)=u(x/t)$ is reported in Figure \ref{fig-late-imp}.

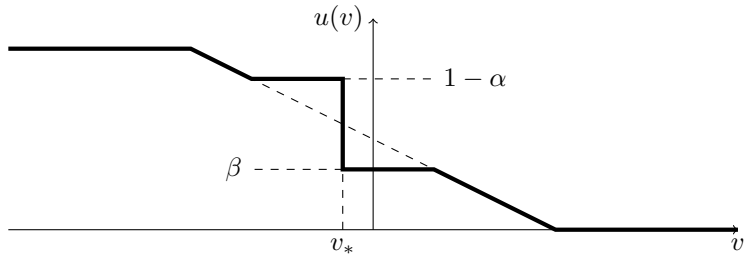
\begin{figure}[H]
\centering
\begin{tikzpicture}[scale = 0.8]
\draw [->] (-5,0)--(7,0) node [below] {$v$};
\draw [->] (1,0)--++(0,3.5) node [left] {$u(v)$};
\draw [ultra thick] (-5,3)--++(3,0)--++(1,-.5)--++(1.5,0)--++(0,-1.5)--++(1.5,0)--++(2,-1)--++(3,0);
\draw [dashed] (-2,3)--++(6,-3);
\draw [dashed] (.5,0) node [below] {$v_\ast$}--++(0,2);
\draw [dashed] (-5,3)++(3,0)++(1,-.5)++(1.5,0)--++(1.5,0) node [right] {$1-\alpha$};
\draw [dashed] (-5,3)++(3,0)++(1,-.5)++(1.5,0)++(0,-1.5)--++(-1.5,0) node [left] {$\beta$};
\end{tikzpicture}
\caption{Late time density profile $\rho(x,t)=u(x/t)$ in the case $\alpha+\beta<1$, for an impurity initially at the interface between a region of density equal to $1$ (left) and a region of density equal to $0$ (right).}\label{fig-late-imp}
\end{figure}

\subsection{General $\rho_L-\rho_R$ step initial density profile}

We now consider the general scenario where the initial condition consists 
of two adjacent regions of uniform density, separated by an impurity.
The density to the right of the impurity is  $\rho_L$, 
while the density to the left is $\rho_R$.
In the previous section, we have seen that if 
 $\alpha>1$ and the sites to the left are all empty, the impurity can 
 escape to the right.  Similarly, if $\beta>1$ and the sites to the right 
 are all occupied, the impurity can escape to the left.
For the remainder of this discussion, we will assume
$ 0 < \rho_L$  and  $\rho_R < 1$. This assumption ensures that we avoid the escaping impurity phenomenon, which has already been analyzed.
We treat separately the cases where the left density $\rho_R$ is larger or smaller than the the right density $\rho_L$.

\subsubsection{Impurity in the rarefaction fan: $ \rho_L>\rho_R$}
\label{sect:nu-gret-mu}

Let's consider the case for which $\rho_L>\rho_R$. In the absence of the impurity, this situation leads to the formation of a simple rarefaction fan. The final result of our analysis is summarized in the diagram in Figure \ref{zoo-fan}.
For the analysis, we again  distinguish the two sub--cases based on 
whether $\alpha + \beta$ is larger or smaller than $1$.

\subsubsection*{Case $\alpha+\beta\geq 1$}

In this case, the particle will not disturb the density 
profile. Its asymptotic speed can be either deterministic or random, 
belonging to a specific interval.
 By an analysis similar to the one we performed
for the case $1-0$, we find that the asymptotic speed of the particle 
depends on the region to which $\rho_L$ and $\rho_R$ belong--either (a), 
(b), or (c).

If both  $\rho_R$ and $\rho_L$ belong to  region (c) (i.e. $1-\alpha>\rho_L>\rho_R$), 
then the impurity moves slower than the left edge of the fan, therefore it will end up in a region with density $\rho_L$ and have a speed equal to  $\alpha - \rho_L$.
By a similar reasoning we conclude that if 
both $\rho_R$ and $\rho_L$ belong to region (b) (i.e. $\rho_L>\rho_R>\beta$), then 
the impurity moves faster than the right edge of the fan, and therefore it 
will end up in a region with density $\rho_R$ and have a speed equal 
to  $1-\beta - \rho_R$.
In the remaining cases, if the impurity finds 
itself in either region (c) or region (b), it will move respectively 
slower or faster than the fan and eventually end up within the fan. 
Once there, its speed coincides with the speed of the characteristic at 
density $\rho \in [\rho_R,\rho_L] \cap [1-\alpha,\beta]$. Since $v_\ast=1-2\rho$, we obtain:
\begin{equation}
v_\ast \in [\max(1-2\rho_L, 1-2\beta) , \min(2\alpha -1, 1-2\rho_R)].
\end{equation}

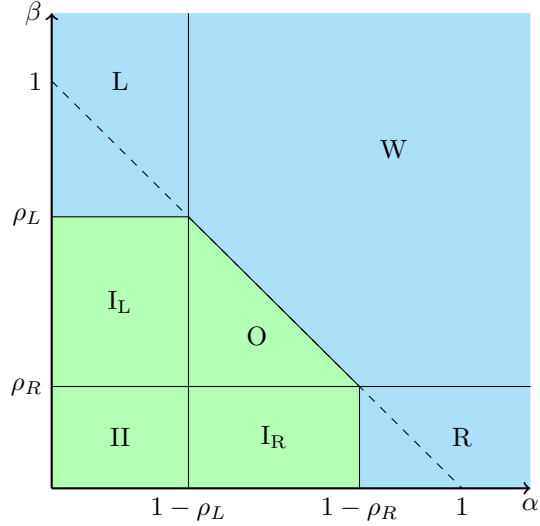
\begin{figure}[h!]
\centering
\begin{tikzpicture}[scale = 0.9]
\fill[cyan!30] (0,7)--(0,4) node [black, left] {$\rho_L$}--(2,4)--(4.5,1.5)--(4.5,0) node[black, below] {$1-\rho_R$}--(7,0)--(7,7)--cycle;
\fill[green!30] (0,4)--(2,4)--(4.5,1.5)--(4.5,0)--(0,0)--cycle;
\draw[thick,->] (0,0) -- (7,0) node[below] {$\alpha$};
\draw[thick,->] (0,0) -- (0,7) node[left] {$\beta$};
\draw[dashed] (0,6) node[left] {$1$}-- (6,0)node[below] {$1$}; 
\draw (2,0) node[below] {$1-\rho_L$} --(2,7);
\draw (0,1.5) node[left] {$\rho_R $} --(7,1.5);
\draw (0,4)--(2,4)--(4.5,1.5)--(4.5,0);
\draw (5,5) node {W};
\draw (1,6) node {L};
\draw (6,.75) node {R};
\draw (3,2.25) node {O};
\draw (1,2.75) node {I$_\textrm{L}$};
\draw (3.25,.75) node {I$_\textrm{R}$};
\draw (1,.75) node {II};

\end{tikzpicture}

\caption{
The diagram describes various potential outcomes for the case $\rho_L>\rho_R$.
The cyan region indicates a fan-like density profile unaltered by the impurity's presence. 
In region 
(W), the impurity is located within the fan, with a speed that takes a random value in the interval  $[\max(1-2\rho_L, 1-2\beta) 
, \min(2\alpha -1, 1-2\rho_R)]$. In region (L) the impurity is located 
to the left of the fan, moving with speed $\alpha-\rho_L$, while in region 
(R) the impurity is located to the right of the fan, moving  with 
speed $1-\beta-\rho_R$. 
In the green region, an anti-shock forms at the impurity's location, moving at speed $\alpha-\beta$. 
Specifically, region (O) features the impurity inducing an anti-shock. Region (I$_\textrm{L}$)
sees the anti-shock preceded by a shock, while in region (I$_\textrm{R}$)
it is followed by a shock. 
Region (II) has the anti-shock both preceded and followed by shocks. The results for the green region are exemplified in the simulation shown in Figure \ref{fig:munu1}.
}\label{zoo-fan}
\end{figure}

\subsubsection*{Case $\alpha+\beta < 1$}

In the case $\alpha+\beta < 1$, the presence of the impurity 
may alter the density profile.
Both the density profile and the asymptotic speed of the impurity depend on which regions $\rho_L$ and $\rho_R$ belong to. 
If both  $\rho_R$ and $\rho_L$ belong to  region (c) (in this case $
\beta>\rho_L>\rho_R$), the impurity does not alter the density profile. It 
moves slower than the left edge of the fan, 
eventually settling in the region with density
$\rho_L$ and moving at speed $\alpha - \rho_L$.
Similarly, if 
both $\rho_R$ and $\rho_L$ belong to region (b) (in this case $\rho_L>\rho_R>1-
\alpha$), the impurity again does not modify the density profile. 
It moves faster than the right edge of the fan, 
eventually settling in the region with density $\rho_R$ and moving at speed  $1-\beta - \rho_R$.
In the remaining cases, regardless of the density surrounding the 
impurity at a given time, the impurity will ultimately end up in region 
(d), leading to the formation of an anti-shock at the impurity's 
location, moving at speed 
$\alpha-\beta$. 
The density to the left of the impurity is $1-\alpha$, so if
$\rho_L<1-\alpha$, the regions at density $\rho_L$ and $1-\alpha$ are separated by a shock; otherwise they are separated by a fan.
Similarly, the density to the right of the impurity is $\beta$, so if
$\rho_R>\beta$ the regions at density $\rho_R$ and $\beta$ are separated by a shock; otherwise they are separated by a fan.
All this possibilities are exemplified in the simulations shown in Figure \ref{fig:munu1d}.

\begin{figure}[h!]
	\centering
	\begin{subfigure}[b]{0.4\linewidth}
		\includegraphics[width=\linewidth]{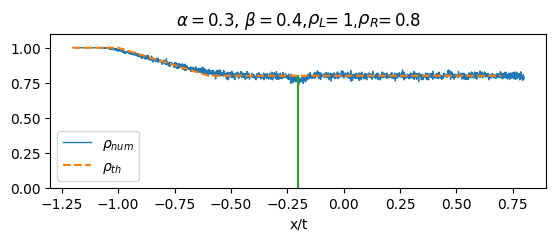}
		\caption{}
		\label{fig:munu1a}
	\end{subfigure}
	\begin{subfigure}[b]{0.4\linewidth}
		\includegraphics[width=\linewidth]{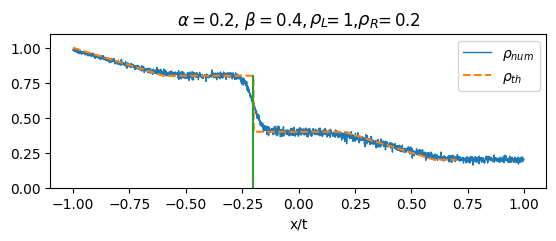}
		\caption{}
		\label{fig:munu1b}
	\end{subfigure}
	\begin{subfigure}[b]{0.4\linewidth}
		\includegraphics[width=\linewidth]{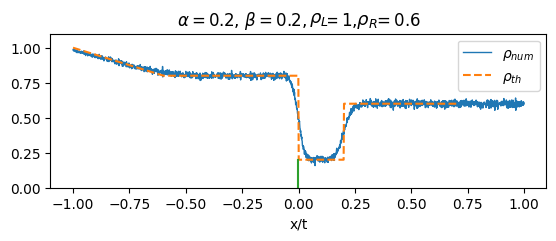}
		\caption{}
		\label{fig:munu1c}
	\end{subfigure}
	\begin{subfigure}[b]{0.4\linewidth}
		\includegraphics[width=\linewidth]{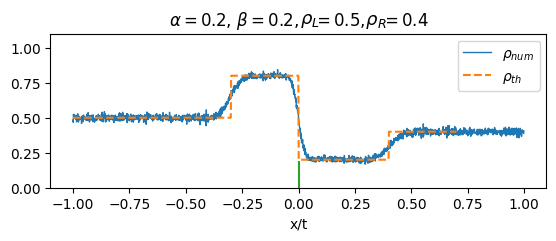}
		\caption{}
		\label{fig:munu1d}
	\end{subfigure}
\caption{Examples of the phenomenology of the density field in the case 
of $\rho_L > \rho_R$ and $\alpha+\beta < 1$. The dashed orange line represents 
the theoretical profile of the density. The blue line represents the 
numerically simulated density averaged for $1000$ realizations run up to 
time $t=1000$. The vertical green line is the simulated position of the 
impurity. In  \ref{fig:munu1a}, the parameters belong to region 
(R) of  \ref{zoo-fan}. Here, we observe the impurity situated to the right of the fan. 
In  \ref{fig:munu1b}, the parameters are in region 
(O), showing an anti-shock within the fan.
In \ref{fig:munu1c}, the parameters belong to region 
(I$_\text{R}$), with an anti-shock followed by a shock. 
Finally, in  \ref{fig:munu1d}, the parameters are in region 
(II), where we see an anti-shock both preceded and followed by a shock.
}
	\label{fig:munu1}
\end{figure}

\subsubsection{Impurity in the shock:  $\rho_L<\rho_R$ }
\label{sect:nu-smaller-mu}

In this case, in absence of an impurity, the density profile would present a shock.
If an impurity is present, there are four different possible scenarios.
Either the impurity moves at constant speed, faster or slower than the 
shock. A third possibility is that the impurity gets stuck at the 
position of the shock without modifying it. The final possibility is that 
the impurity creates an anti--shock possibly followed and/or preceded by shocks.
Let's recall that the speed of the shock is $v_s=1-\rho_L-\rho_R$ and that $v_\ast$ is a weakly decreasing function of the density, which implies that  $v_\ast(\rho_L) \geq v_\ast(\rho_R)$.
We will distinguish four possibilities:
\begin{itemize}[nosep,
  align=left,
  leftmargin=0pt,
  labelwidth=1.25em,
  itemindent=1.25em,
  labelsep=0pt]
\item[$\blacktriangleright$] $\alpha>1-\rho_R$ and $\beta<\rho_L$, then by direct computation we have  
$v_\ast(\rho_R)> v_s$, 
which means that in any case the impurity moves faster than the shock and 
in particular eventually it will be located to its right with speed 
equal to $v_\ast(\rho_R)=1-\beta-\rho_R$.
\item[$\blacktriangleright$] $\alpha<1-\rho_R$ and $\beta>\rho_L$. In this case then $v_\ast(\rho_L)< v_s$, which means that the 
impurity moves slower than the shock and eventually will be located is 
located to its left with a speed equal to 
$v_\ast(\rho_L)=\alpha-\rho_L$.
\item[$\blacktriangleright$] $\rho_L>1-\alpha>\beta>\rho_R$. In this case we have $v_\ast(\rho_L)> v_s>v_\ast(\rho_R)$.
This implies that if the impurity is initially located to the left of the 
shock, it will eventually catch up to the shock since it moves faster. 
Conversely, if the impurity is initially located to the right of the 
shock, the shock will eventually catch up to the impurity since the 
impurity moves slower. It remains to address whether the presence of the 
impurity deforms the shock profile. We already know that deformation 
cannot occur when $\alpha+\beta>1$, but this also holds true even if
$\alpha+\beta<1$. 
Indeed, the only possible deformation would be an anti--shock between 
densities  $1-\alpha$ and $\beta$, preceded by a shock between densities $\rho_L$ and $1-\alpha$, and followed by another shock between densities 
$\beta$ and $\rho_RL$. 
However, this is not feasible since the shock to the left has a speed greater than that of the anti-shock, while the shock to the right has a speed smaller than that of the anti-shock.
In conclusion,  the impurity in this case is located at the position of the shock and thus has a speed equal to  $v_s=1-\rho_L-\rho_R$.
\item[$\blacktriangleright$] $\beta<\rho_L<\rho_R<1-\alpha$. In this case the density surrounding the impurity falls into region (d), which means that 
the density profile consists of an  
  anti--shock between 
densities  $1-\alpha$ and $\beta$, preceded by a shock between densities $\rho_L$ and $1-\alpha$, and followed by another shock between densities 
$\beta$ and $\rho_R$. Contrary to the previous scenario, in this case, the speeds of the shocks and the anti-shock are correctly ordered. The impurity is located at the anti-shock and has a speed of $\alpha-\beta$.
\end{itemize}

The four scenarios discussed above are summarized in Figure \ref{fig-zoo-shock}, while examples of simulations are reported in Figure \ref{shocks}.

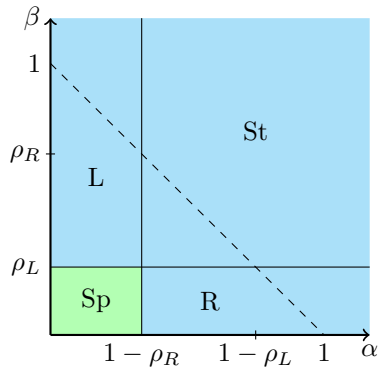
\begin{figure}
\centering
\begin{tikzpicture}[scale = 0.6]
\fill[cyan!30] (0,7)--(0,1.5) --(2,1.5)--(2,0)--(7,0)--(7,7)--cycle;
\fill[green!30] (0,1.5)--(2,1.5)--(2,0)--(0,0)--cycle;
\draw[thick,->] (0,0) -- (7,0) node[below] {$\alpha$};
\draw[thick,->] (0,0) -- (0,7) node[left] {$\beta$};
\draw[dashed] (0,6) node[left] {$1$}-- (6,0)node[below] {$1$}; 
\draw (2,0) node[below] {$1-\rho_R$} --(2,7);
\draw (0,1.5) node[left] {$\rho_L $} --(7,1.5);
\draw (0,4) node [black, left] {$\rho_R$} (4.5,0) node[black, below] {$1-\rho_L$} (4.5,4.5) node {St} (3.5,0.75) node {R} (1,3.5) node {L} 
 (1,.75) node {Sp};
\draw (-.1,4)--++(.2,0) (4.5,-.1)--++(0,.2);
\end{tikzpicture}
\caption{ The diagram describes various potential outcomes for the case $\rho_L<\rho_R$.
The cyan region indicates a shock density profile unaltered by the impurity's presence.
In region (St), the impurity is located at the shock (see Figure \ref{shocksa} for a simulation).
In region (L) the impurity is located  
to the left of the shock and moves with speed $\alpha-\rho_L$, smaller than 
the shock speed (see Figure \ref{shocksc} for a simulation).
In region (R) the impurity is located to the right of the shock and 
moves with speed $1-\beta-\rho_R$, larger than the shock speed (see Figure \ref{shocksb} for a simulation).
Finally, in the green region (Sp)
the shock splits into two shocks separated by an anti--shock 
located at the impurity position and moves at a speed $\alpha - \beta$. The shock on its right moves at a speed $1-\rho_R-\beta$. The shock on the left has a speed $\alpha - \rho_L$ (see Figure \ref{shocksd} for a simulation). 
}\label{fig-zoo-shock}
\end{figure}

\begin{figure}[h!]
	\centering
	\begin{subfigure}[b]{0.4\linewidth}
		\includegraphics[width=\linewidth]{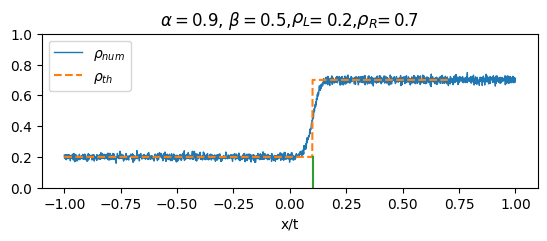}
		\caption{St}
		\label{shocksa}
	\end{subfigure}
	\begin{subfigure}[b]{0.4\linewidth}
		\includegraphics[width=\linewidth]{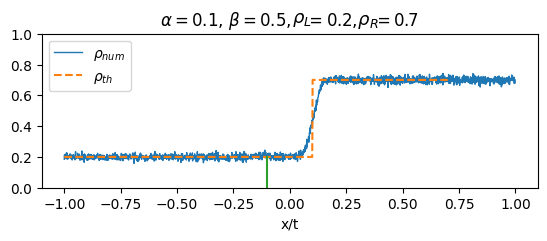}
		\caption{L}
		\label{shocksb}
	\end{subfigure}
	\begin{subfigure}[b]{0.4\linewidth}
		\includegraphics[width=\linewidth]{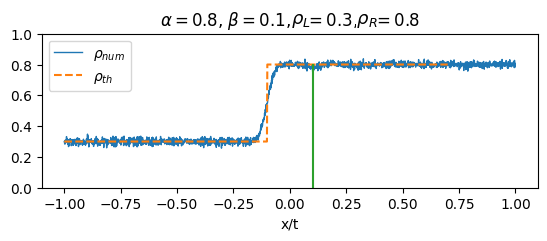}
		\caption{R}
		\label{shocksc}
	\end{subfigure}
	\begin{subfigure}[b]{0.4\linewidth}
		\includegraphics[width=\linewidth]{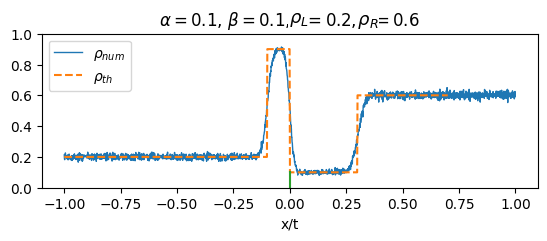}
		\caption{Sp}
		\label{shocksd}
	\end{subfigure}
	\caption{Examples of the phenomenology in the case of $\rho_L < \rho_R$ and $\alpha+\beta < 1$.The dashed orange line represents the theoretical profile of the density. The blue line represents the numerically simulated density averaged for $1000$ realizations run up to 
time $t=1000$. The vertical green line is the simulated position of the impurity particle. The acronyms of the sub-figures refer to Figure \ref{fig-zoo-shock}.}
	\label{shocks}
\end{figure}

\section{Behavior of an impurity}
\label{Asymptotic}

In the previous section we mainly made an analysis of the behavior of
the density profile in  a system with a  
single impurity, which was based on the assumption of the validity of the hydrodynamic limit of this model. In the present section we want to present some rigorous results on the  behavior of the impurity.

To that end, we reframe and extend an approach introduced by Ferrari and 
Pimentel \cite{ferrari2005competition} for representing a second--class 
particle through a hole--particle pair. 
In the following section, as a first step, we formulate this 
representation for an arbitrary deterministic evolution of a particle 
system with an impurity, provided that it satisfies some minimal assumptions that will be 
specified. 
%
%
%
%
In that regard, no distinction is necessary between a second--class particle and an impurity in this framework. 	
In Section \ref{section:dynamics} we "turn on" the stochastic dynamic, and in order to account for our problem, we will need to introduce new constructions that indeed modify the original approach of Ferrari and 
Pimentel.

\subsection{Pair representation of the impurity}


The type of time evolution we consider for a totally asymmetric particle system with exclusion 
and a single impurity is described by a function
 $\eta_*(t,i)$,  $\eta_*:\mathbb{R}^+\times \mathbb{Z}\rightarrow 
\{\circ,\bullet,\ast\}$, which has the following properties:
\begin{itemize}
\item [i.] There exists a unique position $y(t)\in \mathbb{Z}$ such that $\eta_*(t,y(t))=\ast$, representing the location of the impurity.
\item [ii.]  As a function of the time variable $t$,  $\eta_*(t,i)$ is right--continuous with left limits (c\`adl\`ag).
\item[iii.] At any discontinuity point $\tilde{t}$, there is a unique
$j\in \mathbb{Z} $ such that the following conditions hold:
\begin{equation}
\begin{aligned}
\eta_\ast(\tilde{t}^+,i)&=\eta_\ast(\tilde{t},i)& i&\neq j,j+1\\
\eta_\ast(\tilde{t}^+,j+1)&=\eta_\ast(\tilde{t},j) && \\
\eta_\ast(\tilde{t}^+,j)&=\eta_\ast(\tilde{t},j+1).&&  
\end{aligned}
\end{equation}
where  $\eta_\ast(\tilde{t},j)>\eta_\ast(\tilde{t},j+1)$ under the ordering $\circ<\ast< \bullet$.
\end{itemize}
Property (i) specifies that there is exactly one site at each time 
$t$ where the impurity is located. 
Property (iii) explicitly states that any discontinuity point corresponds 
to an exchange between two neighboring particles, consistent with the 
rules of TASEP with an impurity. Specifically:
$$
\bullet \circ \longrightarrow \circ \bullet\qquad
\bullet \ast \longrightarrow \ast \bullet\qquad
\ast \circ \longrightarrow \circ \ast
$$ 
To associate a particle evolution $\eta_\ast$ with a single impurity to a 
particle evolution $\eta$ without any impurity, we "split" the impurity $
\ast$ into a \emph{hole--particle pair} $\circ \bullet$. This 
construction is expressed mathematically as follows: for a given 
configuration $\eta_\ast(t,\cdot): \mathbb{Z} \rightarrow\{\circ, 
\bullet,\ast\}$ at time 
$t$, define a corresponding configuration $\eta(t,\cdot):\mathbb{Z} 
\rightarrow\{\circ, \bullet\}$  for the system without an impurity.
At any time $t$, identify the unique position $y(t)\in\mathbb{Z}$ such 
that $\eta_\ast(t,y(t))=\ast$. Then, replace the impurity $\ast$ 
at position $y(t)$ with a pair consisting of a hole $\circ$ immediately 
followed by a particle $\bullet$. Formally, define the mapping:
\begin{equation}
\begin{aligned}
\eta(t,y(t))&=\circ  &&\\
\eta(t,y(t)+1)&=\bullet &&\\   
\eta(t,i)&=\eta_\ast(t,i) &i&<y(t)\\
\eta(t,i)&=\eta_\ast(t,i-1)& i&>y(t)+1
\end{aligned}
\end{equation}
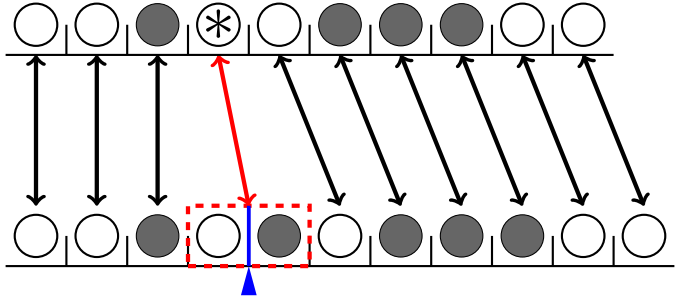
\begin{figure}
\begin{center}
\begin{tikzpicture}[scale = 0.4]
\draw [ultra thick, red,<->]  (7,0)-- (8,-5);
\foreach \t in {1,3,5}
{
\draw [ultra thick, <->]  (\t,0)-- (\t,-5);
}
\foreach \t in {9,11,...,19}
{
\draw [ultra thick, <->]  (\t,0)--++ (2,-5);
}
\begin{scope}[xshift=0cm]
\draw[thick] (0,0) -- (20,0);
\foreach \t in {2,4,...,18}
{
\draw [thick] (\t,0)--++(0,1);
}
\draw [fill=black!60]  (5,1) circle (.7)  (13,1) circle (.7) (15,1) circle (.7) (11,1) circle (.7);
\draw (7,1) node {\Huge $\ast$};
\draw [thick]  (1,1) circle (.7) (3,1) circle (.7) (7,1) circle (.7) (9,1) circle (.7) (19,1) circle (.7) (17,1) circle (.7);
\end{scope}
\begin{scope}[yshift= -7cm]
\draw[thick] (0,0) -- (22,0);
\foreach \t in {2,4,...,20}
{
\draw [thick] (\t,0)--++(0,1);
}
\draw [fill=black!60]  (5,1) circle (.7) (9,1) circle (.7) (13,1) circle (.7) (15,1) circle (.7) (17,1) circle (.7);
\draw [thick]  (1,1) circle (.7) (3,1) circle (.7) (7,1) circle (.7) (11,1) circle (.7) (19,1) circle (.7) (21,1) circle (.7);
\draw [ultra thick, dashed, red] (6,0)--++(4,0)--++(0,2)--++(-4,0)--cycle; 
\draw [line width=0.5mm, blue] (8,0)--++(0,2);
\fill [blue] (8,0)++(-75:1)--++(-75:-1)--++(-105:1)--cycle;
\end{scope}
\end{tikzpicture}
\end{center}
\caption{This picture illustrates 
the correspondence between an impurity and a hole--particle pair. In the top 
configuration, the impurity is represented by a {\Large $\circledast$}. In 
the bottom configuration, this impurity is decomposed into a hole--particle 
pair, which we highlight with a dashed red box. The midpoint of the 
hole--particle pair is marked with a blue triangle.}
\end{figure}
The mapping from $\eta_\ast$ to $\eta$ is well--defined 
regardless of conditions (ii) and (iii) above. 
The purpose of these conditions is to ensure that the mapping can be 
inverted, enabling analysis of the system's behavior using tools and 
methods applicable to the standard exclusion process without impurities.
In order to invert the mapping we need just to  know the 
position of the hole-particle pair at some time $t_0$.
Specifically, let $y(t_0)$ and $y(t_0)+1$ be the positions of the hole 
and the particle in the pair at time $t_0$.  Thanks to the properties of 
$\eta$,  we can track the pair along its time evolution.
To do this, we define the c\`adl\`ag function $x_\eta: \mathbb{R}^+ 
\rightarrow \mathbb{Z}+\frac{1}{2}$ whose value represents the "middle point" 
position of the hole--particle pair.  
\begin{itemize}
\item [$\blacktriangleright$] $x_\eta(t_0)=y(t_0)+1/2$;
\item [$\blacktriangleright$] $\eta(t,x_\eta(t)-1/2)=\circ$ and $\eta(t,x_\eta(t)+1/2)=\bullet$;
\item [$\blacktriangleright$] at the discontinuity points $\tilde{t}$ we have $|x_\eta(\tilde{t})-x_\eta(\tilde{t}^+)|=1 $.
\end{itemize}
What happens is that the value of $x_\eta$ increases  by $1$ whenever 
the particle in the pair jumps to the right, and it decreases by  $1$ 
whenever another particle fills the 
hole in the pair.
To reconstruct the particle evolution with an impurity, $\eta_\ast$,  
from $\eta$ and the initial position $x_\eta(t_0)$, 
we simply "contract" the hole-particle pair back into a single impurity:
\begin{equation}
\begin{aligned}
\eta_\ast(t,x_\eta(t)-1/2)&= \ast &&&\\
\eta_\ast(t,i)&= \eta(t,i)&&\text{for}&  i&<x_\eta(t)-1/2\\
\eta_\ast(t,i)&= \eta(t,i+1)&&\text{for} & i&\geq x_\eta(t)
+1/2
\end{aligned}
\end{equation}
In particular, this gives us  $y(t)=x_\eta(t)-1/2$. 
Thus, if we are 
interested in the movement of the impurity $\ast$,
we can focus on the dynamics of $\eta$ and, more specifically, on the 
evolution of  $x_\eta(t)$. 
\begin{figure}
\begin{center}
\begin{tikzpicture}[scale = 0.3]
\begin{scope}
\draw [ultra thick,->] (8,-2)--++(0,-1.5) node[right] {(a)} --++(0,-1.5);
\draw[thick] (0,0) -- (16,0);
\foreach \t in {2,4,...,14}
{
\draw [thick] (\t,0)--++(0,1);
}
\draw [ultra thick, ->]  (9,2.1) to [in=135,out=45] (11,2.1);
\draw [fill=black!60]  (5,1) circle (.7) (9,1) circle (.7) (15,1) circle (.7);
\draw [ultra thick]  (1,1) circle (.7) (3,1) circle (.7) (7,1) circle (.7) (11,1) circle (.7) (13,1) circle (.7);
\draw [ultra thick, dashed, red] (6,0)--++(4,0)--++(0,2)--++(-4,0)--cycle; 
\draw [line width=0.5mm, blue] (8,0)--++(0,2);
\fill [blue] (8,0)++(-75:1)--++(-75:-1)--++(-105:1)--cycle;
\end{scope}
\begin{scope}[xshift= 22cm]
\draw [ultra thick,->] (8,-2)--++(0,-1.5) node[right] {(b)} --++(0,-1.5);
\draw[thick] (0,0) -- (16,0);
\foreach \t in {2,4,...,14}
{
\draw [thick] (\t,0)--++(0,1);
}
\draw [ultra thick, ->]  (5,2.1) to [in=135,out=45] (7,2.2);
\draw [fill=black!60]  (5,1) circle (.7) (9,1) circle (.7) (15,1) circle (.7);
\draw [ultra thick]  (1,1) circle (.7) (3,1) circle (.7) (7,1) circle (.7) (11,1) circle (.7) (13,1) circle (.7);
\draw [ultra thick, dashed, red] (6,0)--++(4,0)--++(0,2)--++(-4,0)--cycle; 
\draw [line width=0.5mm, blue] (8,0)--++(0,2);
\fill [blue] (8,0)++(-75:1)--++(-75:-1)--++(-105:1)--cycle;
\end{scope}
\begin{scope}[yshift= -8cm]
\draw[thick] (0,0) -- (16,0);
\foreach \t in {2,4,...,14}
{
\draw [thick] (\t,0)--++(0,2);
}
\draw [fill=black!60]  (5,1) circle (.7) (11,1) circle (.7) (15,1) circle (.7);
\draw [ultra thick]  (1,1) circle (.7) (3,1) circle (.7) (7,1) circle (.7) (9,1) circle (.7) (13,1) circle (.7);
\draw [ultra thick, dashed, red] (8,0)--++(4,0)--++(0,2)--++(-4,0)--cycle; 
\draw [line width=0.5mm, blue] (10,0)--++(0,2);
\fill [blue] (10,0)++(-75:1)--++(-75:-1)--++(-105:1)--cycle;
\end{scope}
\begin{scope}[xshift= 22cm,yshift= -8cm]
\draw[thick] (0,0) -- (16,0);
\foreach \t in {2,4,...,14}
{
\draw [thick] (\t,0)--++(0,2);
}
\draw [fill=black!60]  (7,1) circle (.7) (9,1) circle (.7) (15,1) circle (.7);
\draw [ultra thick]  (1,1) circle (.7) (3,1) circle (.7) (5,1) circle (.7) (11,1) circle (.7) (13,1) circle (.7);
\draw [ultra thick, dashed, red] (4,0)--++(4,0)--++(0,2)--++(-4,0)--cycle; 
\draw [line width=0.5mm, blue] (6,0)--++(0,2);
\fill [blue] (6,0)++(-75:1)--++(-75:-1)--++(-105:1)--cycle;
\end{scope}

\end{tikzpicture}
\end{center}
\caption{Possible evolution of a hole--particle pair: in case (a), the particle within the hole--particle pair moves to the right, causing the pair as a whole to shift rightward. In case (b), a new particle moves into the hole of the hole--particle pair, resulting in a leftward shift of the pair.
}
\end{figure}
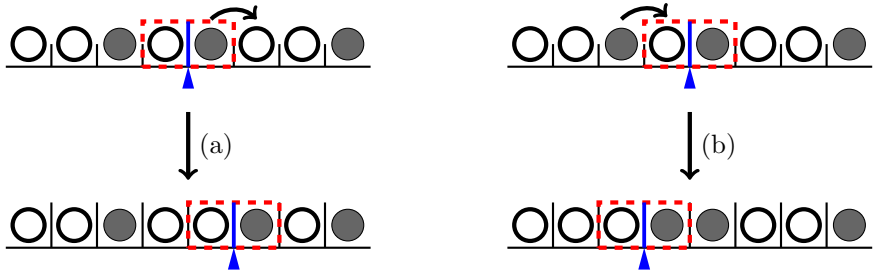
To this end, note that determining $x_\eta(t)$ does not require knowledge 
of the full time evolution $\eta(t,i)$. For example, if we tag a particle
to the right of the hole--particle  pair, the 
displacement of the pair is fully determined by the particle 
configuration to the left of the tagged particle.
Similarly, if we tag a hole
to the left of the hole--particle  pair the 
displacement of the pair is fully determined by the particle 
configuration to the right of the tagged hole.
%
Now, what we described in simpler terms is formally stated in the following Lemma, where $x_\eta^\bullet(t)$ and $x_\eta^\circ(t)$ represent respectively the positions (as functions of time) of a tagged particle and a tagged hole within the particle evolution $\eta$.
\begin{lemma}\label{restriction-lemma}
Let $\eta$ and $\eta'$ be two particle systems having a hole--particle  
pair at the same position at a given time $t_0$, i.e. $x_\eta(t_0)=x_{\eta'}(t_0)$.

\noindent \underline{Case 1: Tagged Particle.}
Suppose that both systems, $\eta$ and $\eta'$,  contain a tagged particle.
Additionally these tagged particles are required follow the same trajectory,  i.e. $x_\eta^\bullet(t)=x_{\eta'}^\bullet(t)$, and at time $t_0$ they  are located to the  right of the h--p pair; i.e. $x_\eta^\bullet(t_0)=x_{\eta'}^\bullet(t_0)>x_\eta(t_0)$.
%
Assume further that
\begin{equation}
\eta(t,i) = \eta'(t,i) \quad \text{for}\quad t\geq t_0, ~~ 
 i \leq x_\eta^\bullet(t) .
\end{equation}

\noindent \underline{Case 2: Tagged Hole.}
Suppose that both systems, $\eta$ and $\eta'$,  contain a tagged hole.
Additionally these tagged holes are required to follow the same trajectory,  i.e. $x_\eta^\circ(t)=x_{\eta'}^\circ(t)$, and at time $t_0$ they  are located to the  left of the h--p pair; i.e. $x_\eta^\circ(t_0)=x_{\eta'}^\circ(t_0)<x_\eta(t_0)$.
Assume further that
\begin{equation}
\eta(t,i) = \eta'(t,i) \quad \text{for}\quad t\geq t_0, ~~ 
 i \geq x_\eta^\circ(t) .
\end{equation}

\noindent
Then in both cases  we have
\begin{equation}
x_\eta(t) = x_{\eta'}(t)\quad \text{for} \quad t\geq t_0.
\end{equation}
\end{lemma}
\begin{proof}
We prove the statement for a tagged particle; the proof for a tagged hole is analogous.
Let $I_\bullet(t)$ be the number of particle between $x_\eta(t)$
and $x_\eta^\bullet(t)$. This number can only increase, 
which occurs whenever the hole-particle pair moves to the left. Since $x_\eta^\bullet(t)-x_\eta(t)>I_\bullet(t)$, and $I_\bullet(t_0)\geq 0$, we conclude that 
$x_\eta(t)<x_\eta^\bullet(t)=x_{\eta'}^\bullet(t)$.
%
Both $x_\eta(t)$ and $x_{\eta'}(t)$ are c\`adl\`ag functions and coincide 
at $t=t_0$.
Therefore, if, by contradiction, they were not equal at all times, there 
must exist
$t_m= \max\{\,t\,|\,x_\eta(t)=x_{\eta'}(t)\}$ such that $x_\eta(t_m^+)
\neq x_{\eta'}(t_m^+)$. For this to occur, a necessary condition is that at least 
one of the following  holds: 
$$\eta(t_m^+,x_{\eta}(t_m)+ 1/
2)\neq \eta'(t_m^+,x_{\eta}(t_m)+ 1/2) \qquad \eta(t_m^+,x_{\eta}(t_m)- 1/
2)\neq \eta'(t_m^+,x_{\eta}(t_m)- 1/2)$$. However, this is not possible 
since $x_{\eta}(t_m)+ 1/2\leq x_\eta^\bullet(t_m)$. 
\end{proof}
The previous lemma essentially states that $x_\eta(t)$ depends only on 
the configuration of the system $\eta$ to the left of a tagged particle 
that is initially in front of the h--p pair, or to the right of a tagged 
hole that is initially to the right of the h--p pair. As a direct 
consequence of this lemma, we obtain the following result:
\begin{cor}\label{restriction-lemma-stoch}
Let $\eta$ and $\eta'$ be two stochastic particle systems 
that coincide at an initial time $t_0$. 
Assume further that both systems exhibit a hole-particle (h--p) pair located at the same position $x_0$ at time $t_0$, so that $x_\eta(t_0) = x_{\eta'}(t_0)=x_0$.

\noindent \underline{Case 1: Tagged Particle.}
Tag the particle belonging to hole--particle (h--p) pair at time 
$t=t_0$, and let $x_\eta^\bullet(t)$ 
denote its position at times $t\geq t_0$
in the realization of process $\eta$,
while $x_{\eta'}^\bullet(t)$ denotes its position in the realization of the process  $\eta'$.
Assume that, as stochastic processes, the position of the tagged particle and the configuration of particles to its left are identical in distribution for both $\eta$ and $\eta'$.
%

\noindent \underline{Case 2: Tagged Hole.}
Tag the hole belonging to the h--p pair at time $t = t_0$ and 
let $x_\eta^\circ(t)$ denote
its position in the realization at times $t\geq t_0$ 
in the realization of process $\eta$,
while $x_{\eta'}^\circ(t)$ denotes its position in the realization of the process  $\eta'$.
Assume that, as  stochastic processes, the position of the tagged hole 
and the configuration of particles to its right are identical in distribution for both $\eta$ and $\eta'$
%

Then in both cases the processes $x_{\eta}$ and $x_{\eta'}$ are identical in distribution on the interval $[t_0,+\infty )$.
\end{cor}

\subsection{Dynamics of the impurity in terms of an auxiliary  TASEP  with non--homogeneous rates}\label{section:dynamics}

We now apply the hole--particle representation of a single impurity to 
our stochastic particle system. 
The scenario where the particle evolution $\eta(t,i)$ follows the usual TASEP dynamics  
is the one originally introduced by Ferrari and Pimentel in \cite{ferrari2005competition}. 
In the 
last--passage representation of the TASEP, the position of the h--p pair 
corresponds to a competition interface. By analyzing the statistical properties of this interface, Ferrari and Pimentel were able to deduce insights about the asymptotic speed of a second-class particle, including its almost sure existence and its distribution.
We now introduce two alternative constructions that extend the Ferrari and Pimentel framework in two distinct ways.

\vspace*{.3cm}
\noindent 
{\bf First construction.}

\noindent
The first construction leads to our original stochastic model with a single impurity. 
The system follows the standard TASEP dynamics, with the exception that 
whenever a particle is located within the pair (inside the red box), it 
jumps at rate  $\alpha$.
Similarly, whenever a hole is within the pair, it is overtaken at rate
$\beta$.
%
\begin{center}
\begin{tikzpicture}[scale = 0.3]
\begin{scope}
\draw[thick] (0,0) -- (20,0);
\foreach \t in {2,4,...,18}
{
\draw [thick] (\t,0)--++(0,1);
}
\draw [ultra thick, ->]  (9,2.2) to [in=135,out=45] (11,2.2);
\draw [ultra thick, ->]  (5,2.2) to [in=135,out=45] (7,2.2);
\draw [ultra thick, ->]  (15,2.2) to [in=135,out=45] (17,2.2);
\draw (6,3.3) node {$\beta$} (10,3.3) node {$\alpha$} (16,3.3) node {$1$};
\draw [fill=black!60]  (5,1) circle (.7) (9,1) circle (.7) (15,1) circle (.7) (19,1) circle (.7);
\draw [ultra thick]  (1,1) circle (.7) (3,1) circle (.7) (7,1) circle (.7) (11,1) circle (.7) (13,1) circle (.7) (17,1) circle (.7);
\draw [ultra thick, dashed, red] (6,0)--++(4,0)--++(0,2)--++(-4,0)--cycle; 
\draw [line width=0.5mm, blue] (8,0)--++(0,2);
\fill [blue] (8,0)++(-75:1)--++(-75:-1)--++(-105:1)--cycle;
\end{scope}
\end{tikzpicture}
\end{center}
We call this process $\eta_1$ and $\eta_{\ast,1}$ the associated process in which the h--p pair is transformed into an impurity. We have that
$\eta_{\ast,1}$ obeys the stochastic dynamic of our original model.

\vspace*{.3cm}
\noindent 
{\bf Second construction.} 

\noindent
A second process, denoted by $\eta_2$, 
is defined by  tagging both the particle and the hole that initially form the pair. 
In this construction, particles jump to the right into empty sites, 
similar to the first construction. However, here the tagged particle 
jumps at rate $\alpha$, while 
the tagged hole is overtaken at rate $\beta$. 
%
We denote by $\eta_{\ast,2}$ the associated process in which the h--p 
pair is transformed into an impurity.

\vspace*{.3cm}      

While the two processes $\eta_{\ast,1}$ and $\eta_{\ast,2}$ are differently distributed for generic values of the parameters $\alpha$ and $\beta$, we have the following theorem:     
      
\begin{theorem} \label{first:Theorem}
Let  $\eta_{1}$ and $\eta_{2}$ be the two processes defined above, both 
starting from the same initial configuration. Denote by 
$x_1(t)$ and $x_2(t)$  the positions of the hole--particle pair in
 $\eta_{1}$ and $\eta_{2}$, respectively. The 
two processes, $x_1$ and $x_2$, are identical in distribution.
\end{theorem}
This is one of the main results of the paper, and in the following 
sections, we will explore some of its consequences. Since the position of 
the impurity in our original model is given by 
$y_1(t)=x_1(t)-1/2$, 
Theorem \ref{first:Theorem} allows us to express its distribution 
essentially in terms of that of $x_2(t)$, which is much simpler to study.

To prove Theorem \ref{first:Theorem}, we will employ some classical 
results from queueing theory. Therefore, we begin by briefly recalling 
the connection between queues and the TASEP.
Consider a TASEP in which all the particles and the empty sites are labeled, and 
particle $i$ exchanges its position with the empty site $j$
immediately to its right at a rate $\tau_{i,j}$. 
This process is equivalent to a system of $\cdot/M/1$ queues in tandem.
One way to understand this equivalence is by associating particles with 
servers and empty sites with customers. When a customer is served by a 
server, it moves to the next server, where it is served after all other 
customers already present.  
The service time for customer $j$ by server  $i$ is given by  
$\tau_{i,j}$. 
In the standard TASEP, the rate $\tau_{i,j}=1$. Notice  we can also 
reverse this identification, treating empty sites as servers and 
particles as customers.
Several theorems about TASEP can be derived as direct consequences of results from queueing theory \cite{ferrari1994net}.

A remarkable result by Weber \cite{weber1979interchangeability} states 
that a series of $./M/1$
queues initially empty and arranged in tandem are \emph{ interchangeable}. This means that the distribution of the output process remains the same regardless of the order in which the queues are arranged. In other words, reordering the sequence of the queues does not affect the statistical properties of the departure process from the final queue.

\begin{theorem}[Weber \cite{weber1979interchangeability}]\label{theo:weber}
Consider two independent $./M/1$ queueing 
servers arranged in tandem, with service rates $\alpha_1$ and $\alpha_2$. The system is initially empty and the first queue has some
arrival process $A$, with an arbitrary distribution.
Then, the law of the departure process from
the system (i.e. of the departure process from the second queue) is the
same if the rates $\alpha_1$ and $\alpha_2$ are interchanged.
\end{theorem}
As a direct consequence of Theorem \ref{theo:weber}, we obtain the following result:
\begin{cor}\label{interch-coroll}
\underline{Particle case}. Consider a TASEP with two tagged particles of positions $x^\bullet_1(t)$ and $x^\bullet_2(t)$, and advancing rates $\alpha_1$ and $\alpha_2$, respectively. Suppose that at some initial time $t_0$, the two tagged particles are neighbors, i.e., the position of the tagged particles satisfies $x^\bullet_2(t_0) = x^\bullet_1(t_0) + 1$. Then, for any time $t \geq t_0$, the law governing the position of the tagged particle $x^\bullet_1(t)$ remains the same if the rates $\alpha_1$ and $\alpha_2$ are interchanged. \\
\underline{Hole case}. Consider a TASEP with two tagged holes of positions $x^\circ_1(t)$ and $x^\circ_2(t)$, and receding rates $\beta_1$ and $\beta_2$, respectively. Suppose that at some initial time $t_0$, the two tagged holes are neighbors, i.e., the position of the tagged holes satisfies $x^\circ_2(t_0) = x^\circ_1(t_0) + 1$. Then, for any time $t \geq t_0$, the law governing the position of the tagged hole $x^\circ_2(t)$ remains the same if the rates $\beta_1$ and $\beta_2$ are interchanged.   
\end{cor}

\begin{proof}[Proof of Theorem \ref{first:Theorem}]
Consider the following family of processes, denoted by 
$\eta^{(k,\ell)}$, with 
$\ell,k \in \mathbb{N}$, 
and having the same initial configuration, equal to $\eta_1(t_0)=\eta_2(t_0)$.
To define these processes, we tag, at the initial time, the  $k$-th particle to the left of the hole--particle pair and the   
$\ell$-th hole to the right of the hole--particle  pair.
If $k=0$ 
(respectively, $\ell=0$), we tag the particle (respectively, the hole) that belongs to the hole--particle pair itself.
Initially, we set the jump rate of the particle in the hole--particle pair to be equal to  
$\alpha$,  while that of the hole is set to be 
$\beta$, and all other jumps occur at rate of $1$.
We assume that this setup persists until the tagged particle or the tagged hole becomes part of the pair. From the moment the tagged particle joins the pair, it acquires and retains the jump rate 
$\alpha$ indefinitely. Similarly, from the moment the tagged hole joins the pair, it acquires and retains the jump rate $\beta$ indefinitely.
Observe that the two processes $\eta_1$ and $\eta_2$ correspond 
respectively to $\eta^{(+\infty,+\infty)}$ and $\eta^{(0,0)}$.

Let $x_{\eta^{(k,\ell)}} 
$ denote the position of the hole--particle pair in the process $\eta^{(k,\ell)}
$.
Our goal is to show that the distribution of 
$x_{\eta^{(k,\ell)}}$ is the same for all $k,\ell \in\mathbb{N}$. 
Formally, this is expressed as:
\begin{equation}
x_{\eta^{(k,\ell)}} \overset{\text{law}}{=} x_{\eta^{(k',\ell')}}  \qquad \forall k,k',\ell,\ell' \in \mathbb{N}.
\end{equation}
We proceed by proving that $ \forall k,\ell \in \mathbb{N}_{>0}$:
\begin{align}\label{interch1}
x_{\eta^{(k,\ell)}}  &\overset{\text{law}}{=} x_{\eta^{(k-1,\ell)}}  \\ \label{interch2}
x_{\eta^{(k,\ell)}}  &\overset{\text{law}}{=} x_{\eta^{(k,\ell-1)}} 
\end{align}
This result implies that incrementing 
$k$ or $\ell$ does not affect the law of the trajectory of the 
hole--particle pair, thereby establishing the equivalence of distributions for all $(k,\ell)$.
To begin with, let's consider equation eq.\eqref{interch1} and define the 
time $t_\bullet^{(k)}$ as the moment when the $k$--th tagged
particle 
becomes part of the hole--particle pair. It is clear that the two 
processes
$\eta^{(k-1,\ell)}$ and $\eta^{(k,\ell)}$ 
%
can be  coupled to coincide up to time $t_\bullet^{(k)}$, 
implying that $x_{\eta^{(k,\ell)}}$ and $x_{\eta^{(k-1,\ell)}}$ match  up to that time. 
Our task is to show that these letter processes remain identically distributed even for $t>t_\bullet^{(k)}$. 
By Corollary \ref{interch-coroll} (particle case), we know that the law governing the trajectory of the particle which at time $t_\bullet^{(k)}$ belongs to the hole--particle pair is the same for both $\eta^{(k-1,\ell)}$ and $\eta^{(k,\ell)}$. 
Then, using Corollary \ref{restriction-lemma-stoch}, we deduce that the distribution 
of the position of the hole-particle pair after time $t_\bullet^{(k)}$ is identical  
for both processes $\eta^{(k-1,\ell)}$ and $\eta^{(k,\ell)}$, 
thereby establishing eq.\eqref{interch1}. Similarly, by applying Corollary \ref{interch-coroll} (hole case) in conjunction with Corollary \ref{restriction-lemma-stoch}, we can prove eq.\eqref{interch2}.
\end{proof}

\subsection {Asymptotic behavior of an impurity}
\label{sect:speed-fan}

As a first consequence of Theorem \ref{first:Theorem}, we can derive a 
purely probabilistic proof of the asymptotic speed of an impurity placed 
in a uniform background of particles, as given in equation 
\eqref{def-speed}:
\begin{theorem} 
Consider an impurity moving in a background of particles with uniform 
density $\rho$. The impurity has an asymptotic speed given by 
\begin{equation}\label{asympt-speed-theo}
v_\ast(\rho)= v^+(\rho)-v^-(\rho),
\end{equation}
where $v^+(\rho)=\min (\alpha, 1-\rho )$ and $v^{-}(\rho) = \min (\beta, \rho )$.
\end{theorem}
\begin{proof}
We use Theorem \ref{first:Theorem} 
to relate the trajectory of the impurity in our model to that of a hole--particle pair in a similar system. Given the configuration of the background particles, the tagged particle moves at a stationary speed
 $v^{+} = \min (\alpha, 1-\rho )$. 
This result is a direct consequence of Burke's theorem in queueing 
theory. Thus, the dynamics of this tagged particle is equivalent to 
that of a tagged particle with rate $1$, embedded in a uniform profile 
with density  $\rho_R = 1 - v^{+}$ to its right. 
Similarly, the tagged hole moves backward at an asymptotic rate of 
$v^{-} = \min (\beta, \rho )$, as if it were a hole of rate $1$ and the 
density on its left is $\rho_{L} = v^{-}$. 
Therefore, asymptotically, the hole--particle pair 
behaves like a second--class particle in the TASEP, where all the jump rates are equal to $1$  and the density profile becomes 
asymptotically uniform with density $\rho_R$  to the right of the pair
and density $\rho_L$ to the left of the pair.
If $\rho_L=\rho_R$  then both densities equal 
$\rho$, and we know that a second-class particle in this setup moves with 
an average speed equal to the characteristic speed of the underlying 
Burgers equation, which in this case is $v_\ast=1-2\rho$. 
This result matches the expression in 
eq.\eqref{asympt-speed-theo}.
If $\rho_L\neq\rho_R$, then $\rho_L<\rho_R$ and in that case the particle 
density profile forms a shock between $\rho_L$ and $\rho_R$.
In such cases, it is known  \cite{ferrari1991microscopic} that the second-class particle, and thus in our context the hole--particle pair, will be located at the position of the shock and will therefore move with an asymptotic speed of
$1- \rho_{L}-\rho_{R} = v^{+}-v^{-}$, which coincide with the claimed expression in eq.\eqref{asympt-speed-theo}.
\end{proof}
\noindent
We now turn to the problem of an impurity positioned precisely at the interface between two regions. Our goal is to compare two different systems and demonstrate that the law governing the impurity's trajectory is identical in both cases.\\
\emph{System I}.
In this system, we consider an impurity with jumping rates $\alpha,\beta<1$. The impurity is situated at the boundary between a fully occupied region to the left and a completely empty region to the right.\\
\emph{System II}.
Here, the impurity is an ordinary second--class particles (with $\alpha=\beta=1$). This second-class particle is positioned at the interface between two regions: on the left, particles are distributed according to a product Bernoulli measure with density $\rho_L=\beta$; on the right, they follow a product Bernoulli measure with density $\rho_R=1-\alpha$.    
\begin{theorem}\label{theo:main2}
The trajectory of the impurity in System I and that of the second-class particle in System II have the same distribution.
\end{theorem}
\begin{proof}
To prove the theorem, we reformulate both systems using the hole--particle representation of the impurity.

For System I, by Theorem \ref{first:Theorem}, 
we find that the impurity's trajectory has the same distribution as that 
of a hole--particle pair, with the particle initially in the pair having 
a jump rate of $\alpha$ and the hole a jump rate $\beta$. 
Given the initial configuration, the tagged particle's trajectory is 
simply a Poisson process with speed $\alpha$, while the  
trajectory of the tagged hole follows a Poisson process with speed 
$-\beta$. 

Now, in System II, the trajectory of the second--class particle matches that of a hole--particle pair in a homogeneous TASEP with initial density 
$\rho_L$ to the left of the pair and $\rho_R$ 
to the right. By applying Burke's theorem (see Example 3.2 in \cite{spitzer1970interaction}), we know that the trajectory of the particle in the pair is a Poisson process with rate $1-\rho_R=\alpha$.
Similarly, using the particle--hole symmetry of TASEP, the trajectory of the hole initially in the pair is a Poisson process with rate
$-\rho_L=-\beta$.

Thus, in both systems, the particle and hole in the pair behave identically. By applying Lemma \ref{restriction-lemma}, we conclude that the trajectories of the hole--particle pairs have the same distribution in both cases. Consequently, this holds for the impurity in System I and the second-class particle in System II.
\end{proof}
An example illustrating Theorem \ref{theo:main2} 
is presented in Figure 
\ref{fig:comparison}, which shows the position distributions of the impurity and the second-class particle in their respective systems after a finite time.
\begin{figure}[h!]
	\centering
	\begin{subfigure}[b]{0.45\linewidth}
		\includegraphics[width=\linewidth]{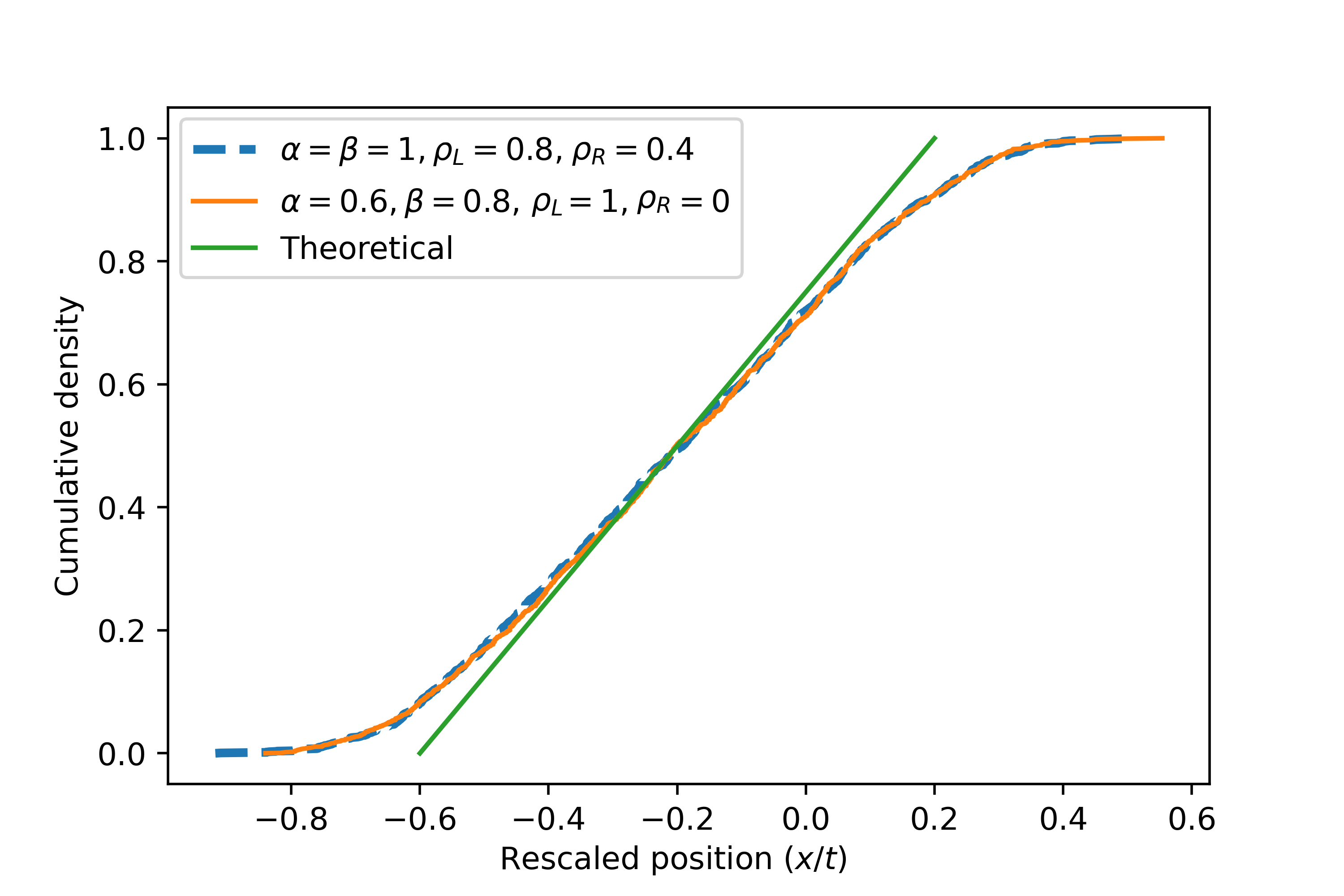}
	\end{subfigure}
\caption{Continuous orange line: sampled cumulative distribution of the rescaled position  
of an impurity particle in a \(1\text{-}0\) step initial profile with \(\alpha = 0.6\) and \(\beta = 0.8\) after time \(t = 500\).  
Dashed cyan line: sampled cumulative distribution of the rescaled position  
of a second-class particle in an initial profile \(\rho_L \text{-} \rho_R\), with \(\rho_L = \beta = 0.8\) and \(\rho_R = 1 - \alpha = 0.4\), after time \(t = 500\).  
Number of realizations in both cases: 2000.}
	\label{fig:comparison}
\end{figure}

\begin{cor}\label{cor-main2}
The asymptotic speed $v_\ast$ of an impurity initially situated at the interface of a $1-0$ profile is as follows: 
\begin{itemize}
\item If $\alpha+\beta< 1$: $v_\ast=\alpha-\beta$.
\item If $\alpha+\beta\geq 1$ and $\alpha,\beta<1$:  $v_\ast$ is 
a random variable uniformly distributed in the interval $[1-2 \beta, 2 \alpha - 1]$.
\end{itemize}

\end{cor}
\begin{proof}
By applying Theorem \ref{theo:main2}, we conclude that the asymptotic speed of the impurity is identical to the asymptotic speed of a second-class particle initially positioned at the interface of a density profile
 $\rho_L-\rho_R$, with $\rho_L=\beta$ and $\rho_R=1-\alpha$.

If $\alpha+\beta<1$, then $\rho_L<\rho_R$,
corresponding to a shock profile that propagates at speed 
$1-\rho_L-\rho_R=\alpha-\beta$. The second-class particle remains at the position of the shock, and therefore moves with speed $\alpha-\beta$.

If $\alpha+\beta\geq 1$, then  $\rho_L\geq \rho_R$. In this case, we apply 
the result from \cite{ferrari1995second}, which establishes that the speed 
is uniformly distributed over the interval 
$[1-2\rho_L,1-2\rho_R]= [1-2\beta,2\alpha-1]$.
\end{proof}
If $\alpha>1$ or $\beta>1$ ,
an impurity initially placed at the interface of a $1-0$ profile has a finite probability 
profile has a finite probability of escaping from the rarefaction fan either from the right or from the left (the  exact probabilities will be
computed in Section \ref{sect:esc}). 
If the impurity escapes to the right, it acquires an asymptotic speed 
of $\alpha$; if it escapes to the left, it acquires an asymptotic speed of  $-\beta$.  Numerical evidence (see Figure \ref{conjecture}) strongly suggests that, for a non--escaping impurity, the asymptotic speed remains uniformly distributed over a suitable interval for any values of 
 of $\alpha$ and $\beta$. 
This is formally stated in the following conjecture, which generalizes Corollary \ref{cor-main2}.
\begin{conj} \label{conj}
For a non--escaping impurity with  $\alpha+\beta \geq 1$,
initially located at the interface of a $1-0$ profile, the probability 
distribution of its asymptotic speed is uniformly distributed over the 
interval $[\max(-1,1-2\beta),\min(1,2\alpha-1)]$.
\end{conj}
For a non--escaping impurity initially situated at the interface of a 
density profile of type $\rho_L-\rho_R$, with arbitrary $\rho_L>\rho_R$,
the analysis in Section \ref{Density profile}  indicates that its asymptotic speed is distributed over the interval
 $[\max(1-2\rho_L, 1-2\beta) 
, \min(2\alpha -1, 1-2\rho_R)]$. However, numerous numerical simulations
(see figure \ref{nonuniform}) suggest that, unless $\alpha=\beta=1$, 
this distribution is not uniform. Determining these non-uniform distributions poses an interesting open problem.

\begin{figure}[h!]
\centering
\includegraphics[width=0.5\linewidth]{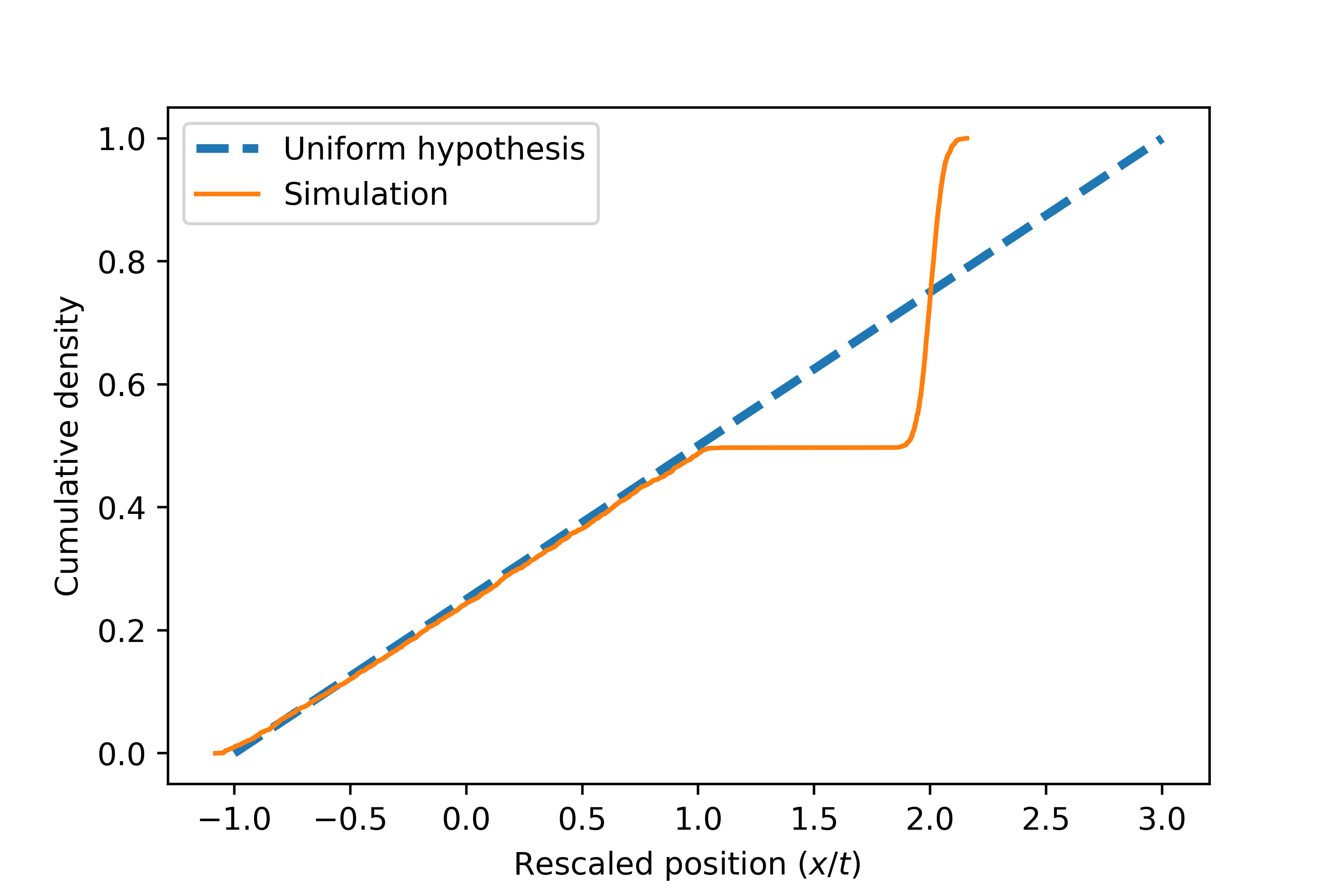}
\caption{Comparison between the simulated cumulative distribution of the asymptotic speed and the uniform distribution, given $\beta= 1, \alpha=2$,
with an initial  $1-0$ particle density profile. 
The simulation was conducted over 4000 realizations, each evolved to $t=1000$. The theoretical uniform distribution is represented across the interval  $[1-2 \beta, 2 \alpha - 1]$, which is the one valid
for $\alpha, \beta < 1$. 
The observed slope supports consistency with Conjecture \ref{conj} and the escape probability formula.}
\label{conjecture}
\end{figure}

\begin{figure}[h!]
	\centering
	\begin{subfigure}[b]{0.4\linewidth}
		\includegraphics[width=\linewidth]{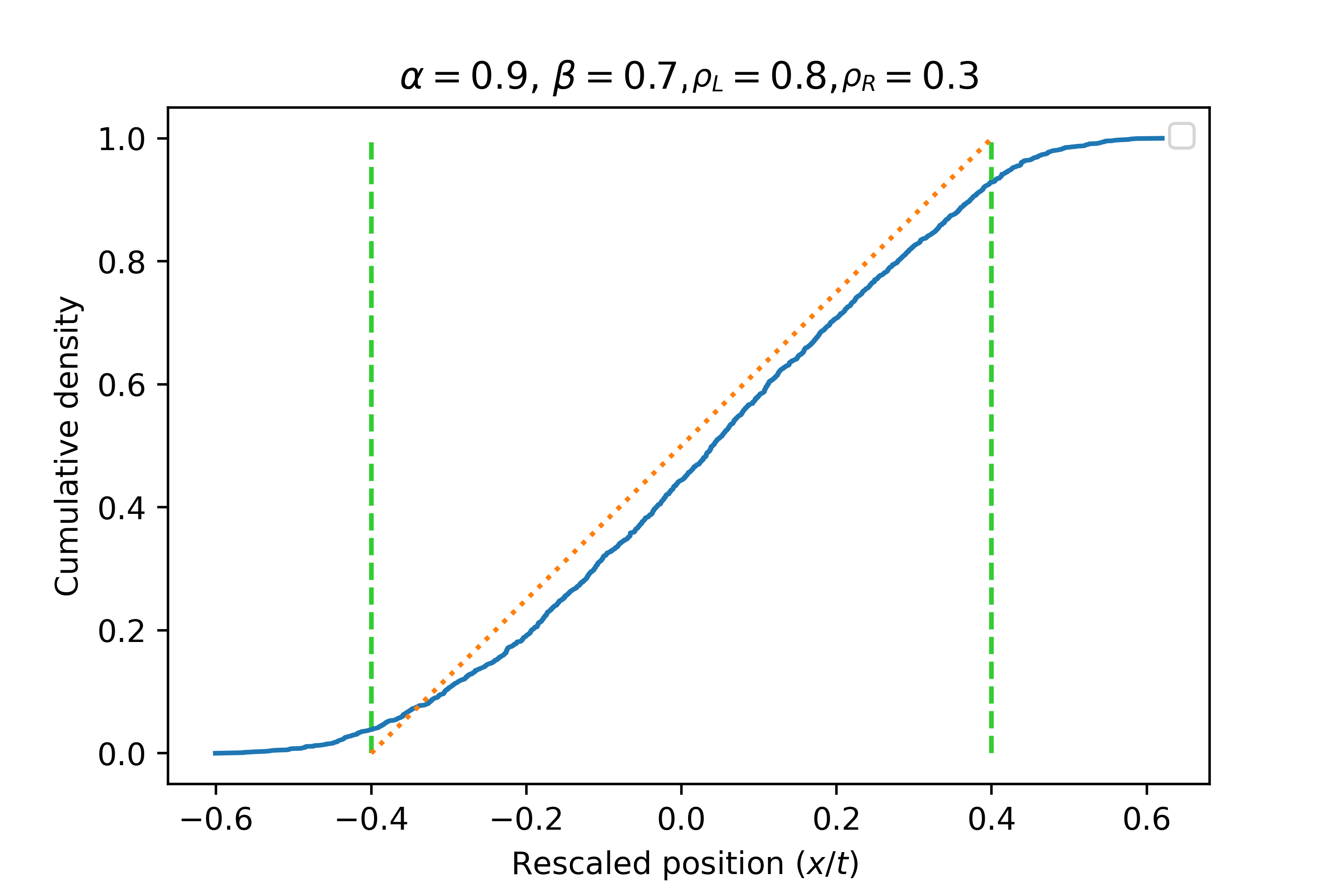}
		\caption{}
		\label{ff}
	\end{subfigure}
	\begin{subfigure}[b]{0.4\linewidth}
		\includegraphics[width=\linewidth]{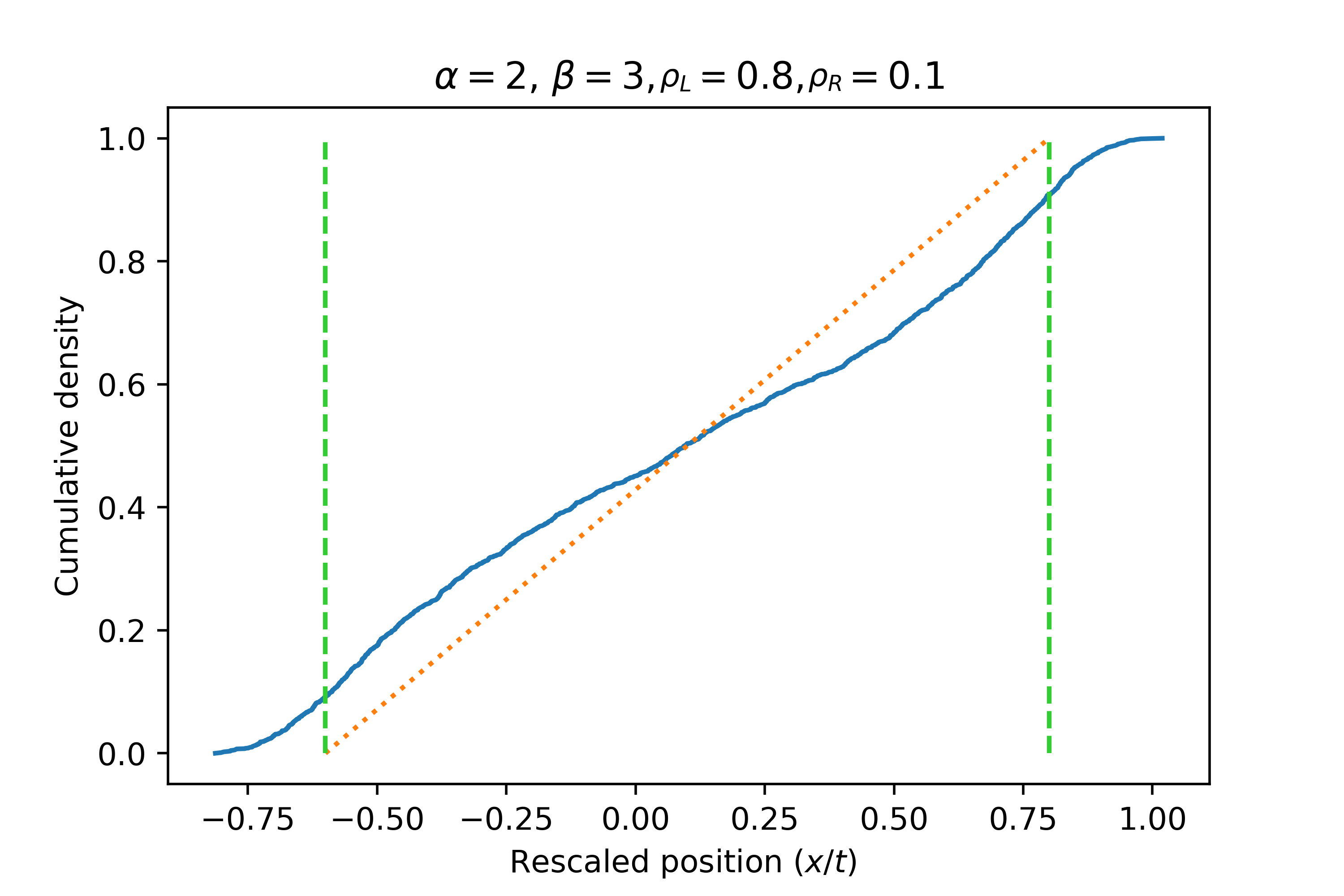}
		\caption{}
		\label{gg}
	\end{subfigure}
	\begin{subfigure}[b]{0.4\linewidth}
		\includegraphics[width=\linewidth]{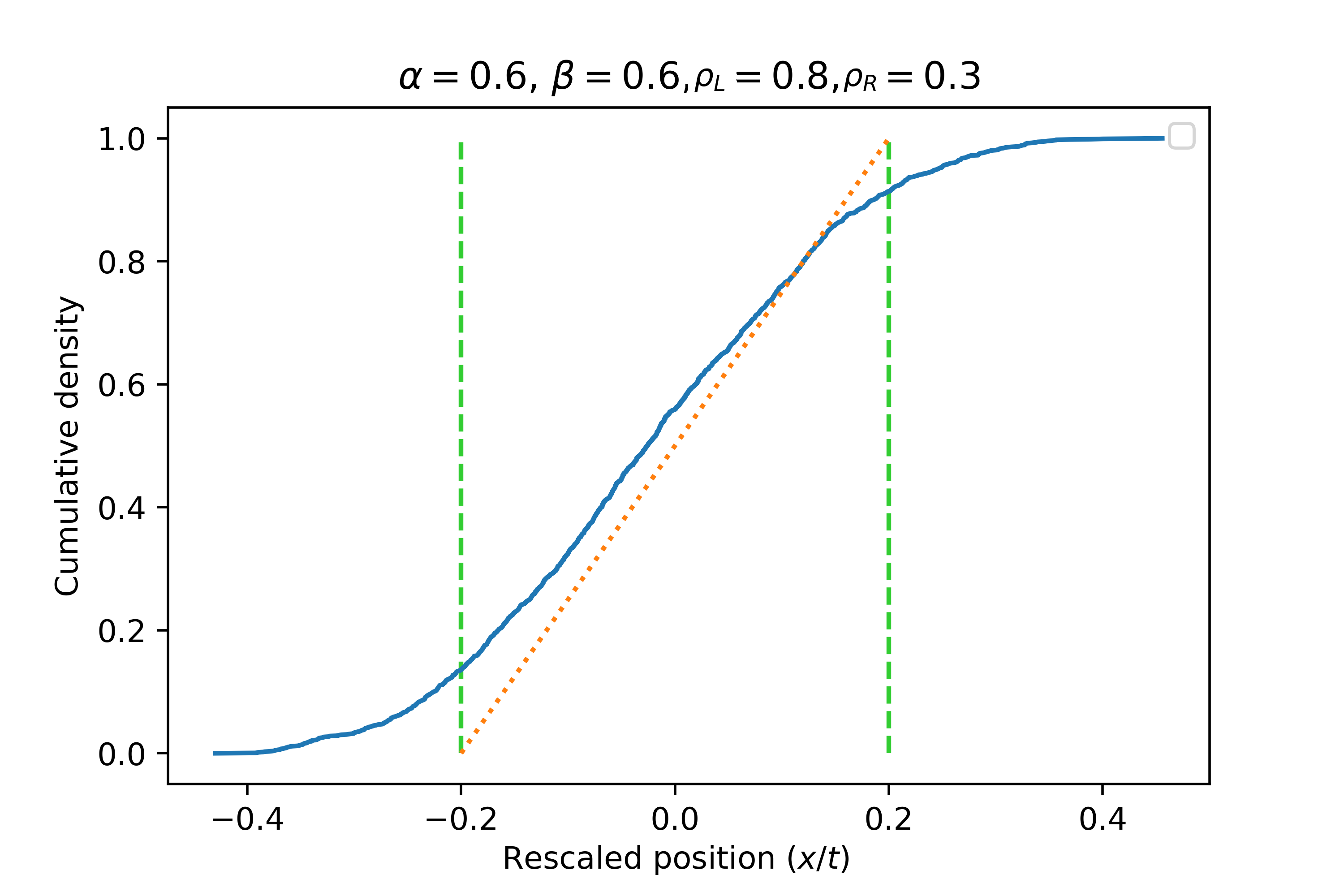}
		\caption{}
		\label{ff2}
	\end{subfigure}
	\begin{subfigure}[b]{0.4\linewidth}
		\includegraphics[width=\linewidth]{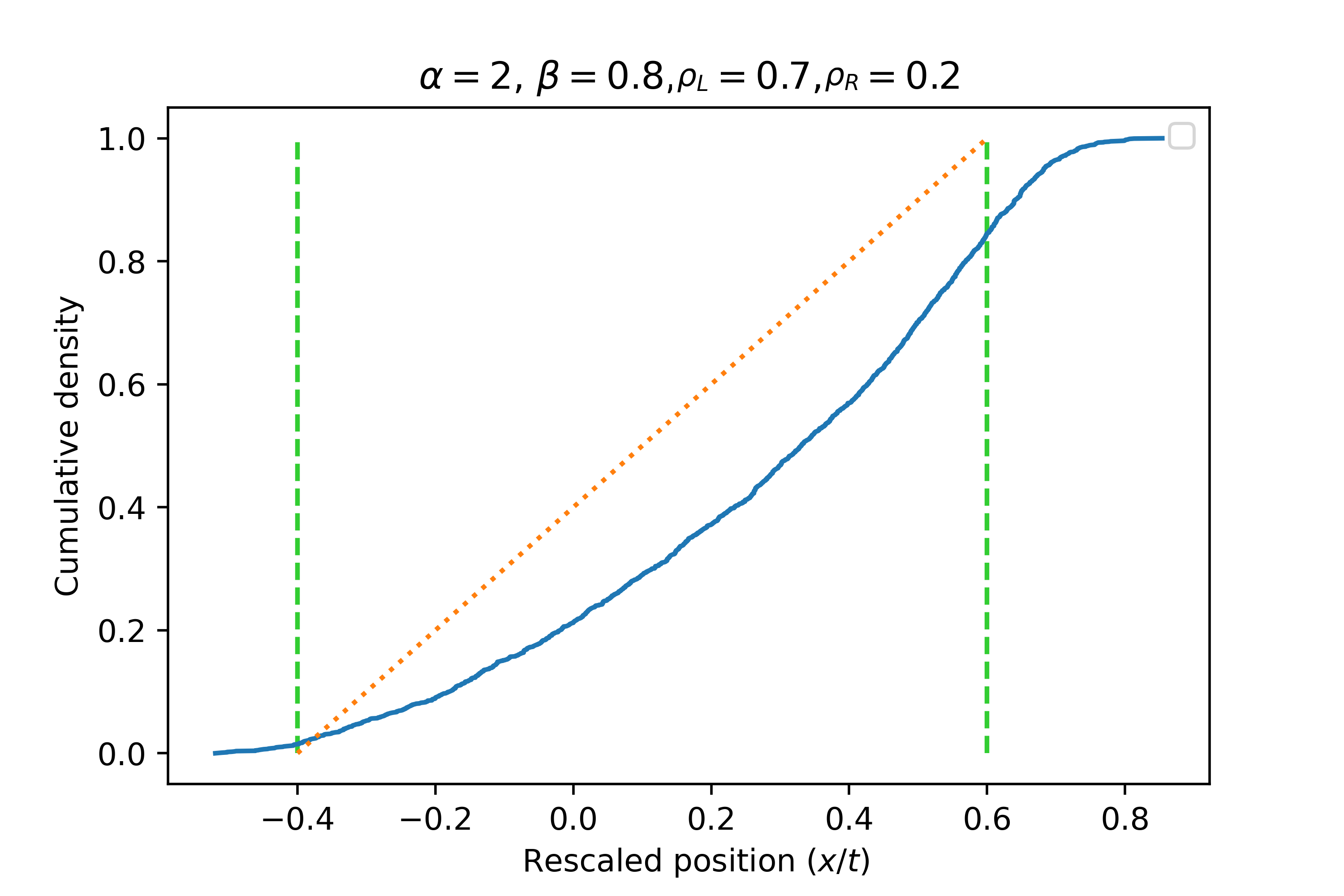}
		\caption{}
		\label{ff3}
	\end{subfigure}
\caption{Examples illustrating nonuniform distributions of the 
asymptotic speed of the impurity particle in cases where $(\rho_R, \rho_L) 
\neq (1, 0)$ and $(\alpha, \beta) \neq (1, 1)$. The predicted range for 
each distribution is shown by vertical dashed lines, while the dotted 
orange lines represent a hypothetical uniform distribution. Each graph is 
based on 1000 realizations.}
	\label{nonuniform}
\end{figure}

\subsection{Escaping probability}
\label{sect:esc}

In this final Section,  we determine the probability that an impurity escapes from a rarefaction fan.
Let's consider an initial condition where all sites to the right of the impurity are empty, while particles are present to its left.  
In this setup, we may ask what is the probability that the impurity is never overtaken by the first particle on its left. 
Since the evolution of a particle is unaffected by other particles positioned further to its left, we can reduce the analysis to a two-particle model involving only the rightmost particle and the impurity.
As long as $\alpha\leq 1$ and 
$\beta>0$, it is straightforward to see that the particle will overtake the impurity with probability $1$.
To understand why, observe that as long as the impurity remains in front of the particle, the distance between them behaves like an asymmetric random walk on $\mathbb{N}_{>0}$.  In this walk, the probability of the distance increasing  is  $\frac{\alpha}{1+\alpha}$, 
while the probability of the distance decreasing is
$\frac{1}{1+\alpha}$.
When $\alpha\leq 1$ the distance  will return to $1$ infinitely often. 
Moreover, since  $\beta>0$, the particle will 
eventually overcome the impurity with  certainty.

When $\alpha > 1$, the impurity has a higher average rightward  speed than 
the particle, resulting in a non--zero probability that the impurity will 
always stay to the right of the particle. We define this probability as the
\emph{right escaping probability},.
If this right escape scenario occurs, the impurity  behaves like a free particle, moving with an asymptotic speed $\alpha>1$.

\begin{figure}[h!]
	\centering
	\begin{subfigure}[b]{0.45\linewidth}
		\includegraphics[width=\linewidth]{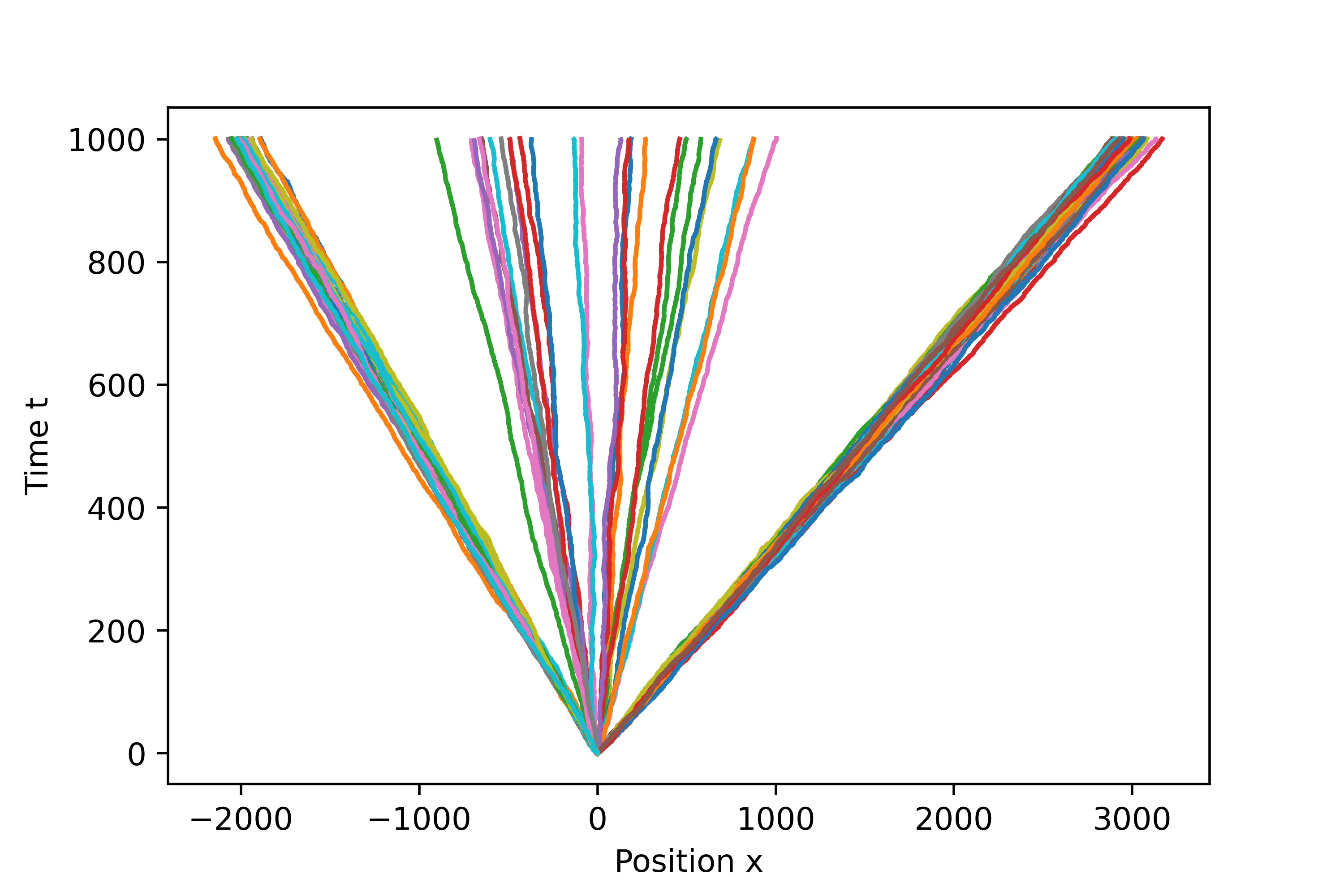}
	\end{subfigure}
	\begin{subfigure}[b]{0.45\linewidth}
		\includegraphics[width=\linewidth]{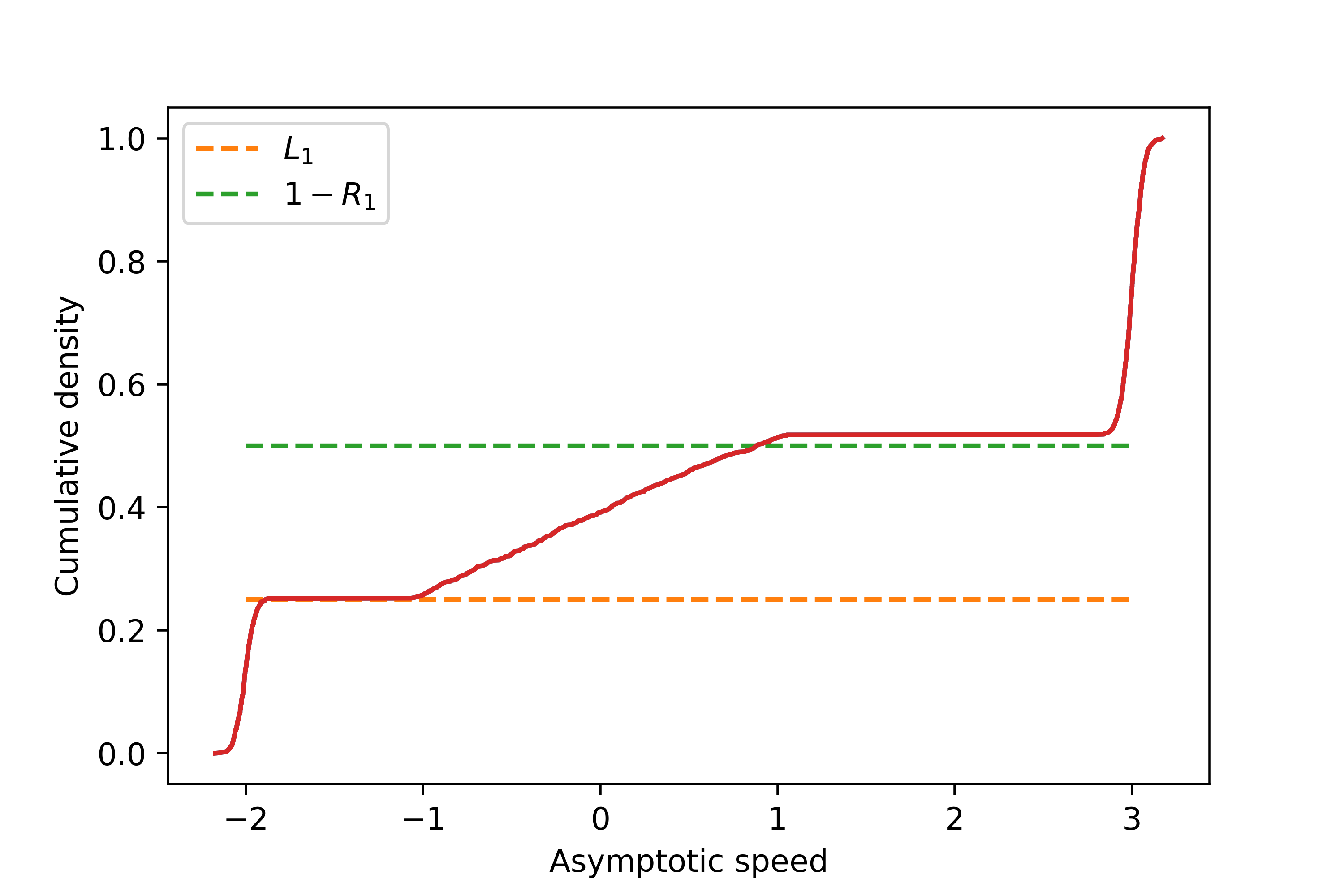}
	\end{subfigure}
	\caption{Trajectories (on the left) and asymptotic speeds cumulative distribution (on the right) of the impurity for $\alpha = 3$ and $\beta = 2$, for a $1-0$ step initial profile. Cumulative density is plotted over 2000 realization for time $t=1000$. Dashed lines represent theoretical values of the escaping probabilities. 100 trajectories are plotted on the left.}
	\label{fig1}
\end{figure}

Let $X_k$ be the random variable representing the distance between the impurity and the particle to its left after a total of $k$ jumps of both the particle and the impurity. The variable $X$ follows an asymmetric random walk with the following rules:
\begin{align*}
&P(X_{k+1} = n+1|X_{k}=n) = \frac{\alpha}{\alpha + 1} := p > \frac{1}{2} \qquad n>1\\
&P(X_{k+1} = n-1|X_{k}=n) = \frac{1}{\alpha + 1} = 1-p \qquad n>1\\
&P(X_{k+1} = 2|X_{k}= 1) = \frac{\alpha}{\alpha + \beta} \\
&P(X_{k+1} = -1|X_{k}= 1) = \frac{\beta}{\alpha + \beta}  
\end{align*}
Our problem corresponds to a random walker with absorbing boundaries \cite{kac1945random}.
Let $$R_{n} = P(X_{l} \geq 1 , \forall\; l >k | X_{0} = n \geq 1)$$
be the right escaping probability 
of a impurity starting from a distance $n$ from the rightmost particle. 
The quantity $R_n$ satisfies the following relations
\begin{align}\label{recursion}
R_{n} &= pR_{n+1} + (1-p)R_{n-1}  \qquad n\geq 2\\\label{recursion2}
R_{1} & = \frac{\alpha}{\alpha+\beta} \,R_{2}
\end{align}
The general solution of the recursion eq.(\ref{recursion}) is of the form
\begin{equation}
R_{n} = A + B \left(\frac{1-p}{p}\right)^{n}=A + B \alpha^{-n} ,
\end{equation} 
where $A$ and $B$ are constants to be determined from initial conditions $R_1$ and $R_2$. 
\begin{equation}
A=\frac{\alpha R_2-R_1}{\alpha-1},\qquad B=\alpha^2 \frac{R_1-R_2}{\alpha-1}.
\end{equation}
In the limit $n\rightarrow \infty$ the escape probability must tend to one, which means that we must have $A=1$.  
Keeping into account that we have also the relation between $R_1$ and $R_2$ of eq.\eqref{recursion2} we determine $R_1$
\begin{equation}
R_{1} =  \dfrac{\alpha - 1}{\alpha + \beta - 1}.
\end{equation}
This is precisely the quantity $P_R$ of eq.(\ref{esc-prob}).
For arbitrary initial distance $n$ we obtain
\begin{equation}
R_{n} = 1 -  \dfrac{\beta}{(\alpha + \beta -1)}\alpha^{1-n}.
\end{equation}
Now suppose that the impurity is at a distance $n$ to the left of a region with a uniform density $\rho_L$ of  particles. We denote the escape probability for such an impurity as $R_n^{(\rho_L-0)}$.
The initial distance between the rightmost particle and the impurity is geometrically distributed as follows:
$$
P_{n,k}=\left\{ 
\begin{array}{cc}
\rho_L(1-\rho_L)^{k-n} & k\geq n\\
0 & k<n
\end{array}
\right.
$$ 
So that for $R_n^{(\rho_L-0)}$ we find
\begin{equation}
R_n^{(\rho_L-0)}=\sum_{k=1}^\infty P_{n,k}R_k= 
1-
\frac{\rho_L\beta\alpha^{2-n}}{(\alpha + \beta -1)(\alpha + \rho_L -1)}.
\end{equation}
Using the particle--hole symmetry ($\alpha \leftrightarrow \beta, \rho_R\leftrightarrow 1-\rho_L$), we obtain as well $L_n^{(1-\rho_R)}$, the left escape probability of a impurity at distance $n$ to the left of a region with a uniform density $\rho_R$ of first-class particles
\begin{equation}
L_n^{(1-\rho_R)} = 
1-
\frac{\alpha(1-\rho_R)\beta^{2-n}}{(\alpha + \beta -1)(\beta -\rho_R)}.
\end{equation}
In particular $P_L$ of eq.(\ref{esc-prob}) is recovered as 
\begin{equation}
P_L=L_1^{(1-0)} =\frac{\beta-1}{\alpha+\beta-1}.
\end{equation}

\section{Conclusion}
The study of a single impurity in the Totally Asymmetric Simple Exclusion Process (TASEP) is compelling for multiple reasons. From a macroscopic standpoint, it exemplifies the impact of a defect -- represented by the impurity -- within a scalar conservation law, where the Burgers equation often serves as a prototypical model.
From a microscopic perspective, the model advances the concept of a 
second--class particle -- typically used just as a theoretical tool -- into a more complex and realistic scenario.

There are several promising directions for future research. First, our new coupling approach may be integrated within the broader framework of last-passage percolation models. In this context, insights into the behavior of competition interfaces studied in works like
\cite{cator2012busemann, cator2013busemann, 
ferrari2005competition,ferrari2006roughening} 
could potentially extend to non--homogeneous environments, opening 
the possibility of calculating the limiting distribution for the asymptotic behavior of an impurity under any general
 $\rho_L-\rho_R$ initial condition, as well as testing our conjecture related to non-escaping impurities.

Another intriguing direction involves exploring interactions between two or more impurities, following a line of inquiry similar to that in \cite{ferrari2009collision} for two second-class particles. This approach could provide deeper insights into the dynamics and collective behaviors in multi--impurity systems within TASEP.

Finally, it would be interesting to extend our analysis to the case of an impurity in a partially asymmetric environment. 
Recent work by Lobaskin and collaborators \cite{lobaskin2022matrix,lobaskin2023integrability,lobaskin2024current}
has identified specific dynamics for the impurity under which the model 
remains solvable. Investigating whether our approach can be adapted to 
this solvable scenario presents a promising direction for further 
research.
%

%

\section*{Acknowledgments}

We thank Patricia Gon\c calves  for reading the manuscript and for her 
valuable remarks and discussions. 
The work of A. Zhara has been partially funded by the ERC
Starting Grant 101042293 (HEPIQ), by FCT/Portugal through project UIDB/04459/2020 with DOI identifier 10-54499/UIDP/04459/2020, and through the grants 2020.03953.CEECIND with DOI identifier 10.54499/2020.03953.CEECIND/CP1587/\\CT0013, 2022.09232.PTDC with DOI identifier 10.54499/2022.09232.PTDC, and the grant number BL155/2023-IST-ID from the Instituto Superior Técnico and FCT (Portugal).

\typeout{}
\bibliography{biblio-gen}
\bibliographystyle{ieeetr}

\end{document}